\documentclass[pdflatex,sn-mathphys-num]{sn-jnl}
\usepackage{graphicx}%
\usepackage{multirow}%
\usepackage{amsmath,amssymb,amsfonts}%
\usepackage{amsthm}%
\usepackage{mathrsfs}%
\usepackage[title]{appendix}%
\usepackage{xcolor}%
\usepackage{textcomp}%
\usepackage{manyfoot}%
\usepackage{booktabs}%
\usepackage{algorithm}%
\usepackage{algorithmicx}%
\usepackage{algpseudocode}%
\usepackage{listings}%
\usepackage{comment}

\usepackage{bussproofs}
\usepackage{tikz-cd}
\usepackage{proof}
\usepackage[title]{appendix}
\newenvironment{bprooftree}
  {\leavevmode\hbox\bgroup}
  {\DisplayProof\egroup}

\newcommand{\Pts}{\sf{Pt-S}}
\newcommand{\Bes}{\sf{Be-S}}
\newcommand{\At}{\sf{At}}
\newcommand{\Atb}{\At_\bot}
\newcommand{\NE}{\sf{NE_B}}
\newcommand{\el}{\mathcal{L}}
\newcommand{\red}[1]{\textcolor{red}{#1}}

\newcommand\ie{\hbox{\textit{i.e.}}}
\newcommand\eg{\hbox{\textit{e.g.}}}

\theoremstyle{thmstyleone}%
\newtheorem{theorem}{Theorem}%  meant for continuous numbers
\newtheorem{proposition}[theorem]{Proposition}% 
\newtheorem{lemma}[theorem]{Lemma}
\newtheorem{corollary}[theorem]{Corollary}
\theoremstyle{thmstyletwo}%
\newtheorem{example}{Example}%
\newtheorem{remark}{Remark}%
\newtheorem{notation}{Notation}

\theoremstyle{thmstylethree}%
\newtheorem{definition}{Definition}%

\begin{document}

\title[An ecumenical view of proof-theoretic semantics]{An ecumenical view of proof-theoretic semantics}

\author[1]{\fnm{Victor} \sur{Barroso-Nascimento}}\email{ victorluisbn@gmail.com}

\author[2]{\fnm{Luiz Carlos} \sur{Pereira}}\email{luiz@inf.puc-rio.br}

\author[1]{\fnm{Elaine} \sur{Pimentel}}\email{e.pimentel@ucl.ac.uk}

\affil[1]{\orgdiv{Computer Science Department}, \orgname{UCL}, \orgaddress{\street{Gower St.}, \city{London}, \postcode{WC1E 6EA},  \country{UK}}}

\affil[2]{\orgdiv{Department of Philosophy}, \orgname{UERJ}, \orgaddress{\street{Rua S\~ao Francisco Xavier}, \city{Rio de Janeiro}, \postcode{20550-900},  \country{Brazil}}}

%\author{Victor Barroso-Nascimento\thanks{Barroso-Nascimento was supported in part by the Coordena\c c\~ao de Aperfei\c coamento de Pessoal de N\'ivel Superior - Brasil (CAPES) - Finance Code 001.} \and Luiz Carlos Pereira\thanks{Pereira is supported by the following projects: CAPES/PRINT, CNPq-313400/2021-0, and CNPq-Gaps and Gluts.} \and Elaine Pimentel\thanks{Pimentel has received funding from the European Union's Horizon 2020 research and innovation programme under the Marie Sk\l odowska-Curie grant agreement Number 101007627. Pimentel and Barroso-Nascimento are supported by the Leverhulme grant RPG-2024-196.} }
%\institute{UERJ/CNPq, Brazil\\University College London, UK}

\abstract{
Debates concerning philosophical grounds for the validity of classical and intuitionistic logics often have the very nature of proofs as a point of controversy. The intuitionist advocates for a strictly constructive notion of proof, while the classical logician advocates for a notion which allows the use of non-constructive principles such as \textit{reductio ad absurdum}. In this paper we show how to coherently combine  {\em logical ecumenism} and  {\em proof-theoretic semantics} ($\Pts$) by providing not only a medium in which classical and intuitionistic {\em logics} coexist, but also one in which their respective {\em notions of proof} coexist. Intuitionistic proofs receive the standard treatment of $\Pts$, whereas classical proofs are given a semantics based on ideas by David Hilbert.  Furthermore, we advance the state of the art in $\Pts$ by introducing a key contribution: treating the absurdity constant $\bot$ as an atomic proposition and requiring all bases to be consistent. This treatment is essential for the obtainments of some ecumenical results, and it can also be used in standard intuitionistic $\Pts$. Additionally, we employ normalization techniques to demonstrate the consistency of simulation bases. These innovations provide fresh technical and conceptual insights into the study of bases in $\Pts$.

}


\maketitle

\section{Introduction}
{\em What is the meaning of a logical connective?} This is a very difficult and controversial question, for many reasons. 
First of all, it depends on the logical setting.  For example, asserting that 
\begin{quote}
$A\vee B$ is valid only if it is possible to give a proof of either $A$ or $B$ 
\end{quote}
clearly does not correctly determine the meaning of the classical disjunction. It turns out, as shown in~\cite{piecha2015failure} and further analysed in~\cite{pym2023categorical,gheorghiu2022prooftheoretic}, that this also does not seem enough for determining meaning in intuitionistic logic, due to the intrinsic non-determinism on choosing between $A$ or $B$ for validating $A\vee B$. 

In model-theoretic semantics, mathematical structures help in supporting the notion of validity, which is based on a notion of {\em truth}. In the case of intuitionistic logic, for example, one could use Kripke structures, where the validation of atomic propositions using the classical notion of truth (\eg\ via truth tables) is enough for describing the meaning of the disjunction {\em in a given world}, where worlds are organised in a pre-order. 

Although it became common to specify the meaning of formulas in terms of truth conditions, we agree with Quine's objection to that, quoting Prawitz~\cite{DBLP:journals/Prawitz15}
\begin{quote}
Following Tarski, he  [Quine] states truth conditions of compound sentences, not as a way to explain the logical constants, but as a first step in a definition of logical truth or logical consequence, which Quine takes to demarcate the logic that he is interested in. He points out that the truth conditions do not explain negation, conjunction, existential quantification and so on, because the conditions are using the corresponding logical constants and are thus presupposing an understanding of the very constants that they would explain. I think that he is essentially right in saying so and that the situation is even worse: when stating truth conditions, one is using an ambiguous natural language expression that is to be taken in a certain specific way, namely in exactly the sense that the truth condition is meant to specify.
\end{quote}

Proof-theoretic semantics~\cite{pts-91,schroeder2006validity,sep-proof-theoretic-semantics} ($\Pts$) provides an alternative perspective for the meaning of logical operators compared to the viewpoint offered by model-theoretic semantics. In $\Pts$, the concept of {\em truth} is substituted with that of {\em proof}, emphasizing the fundamental nature of proofs as a means through which we gain demonstrative knowledge, particularly in mathematical contexts. $\Pts$ has as philosophical background  {\em inferentialism}~\cite{Brandom2000}, according to which inferences establish the meaning of expressions. This makes $\Pts$ a superior approach for comprehending reasoning since it ensures that the meaning of logical operators, such as connectives in logics, is defined based on their usage in inference. 

Base-extension semantics~\cite{Sandqvist2015IL} ($\Bes$) is a strand of $\Pts$ where proof-theoretic validity is defined relative to a given collection $S$ of inference rules defined over basic formulas of the language.\footnote{It should be noted that, in~\cite{Sandqvist2015IL},  base rules are restricted to formulas in the logic-free fragment only, that is, to {\em atomic propositions}. Here we will follow~\cite{piecha2015failure} and give the unit $\bot$ an ``atomic status'', allowing it to appear in atomic rules.} Hence, for example, while satisfiability of an atomic formula $p$ at a state $w$ in a Kripke model $\mathcal{M}=(W,R,V)$ is often given by
\[w\Vdash p \qquad \mbox{iff} \qquad w\in V(p) 
\]
in $\Bes$, validity w.r.t. a set $S$ of atomic rules has the general  shape
\[
\Vdash_{S} p \qquad \mbox{iff} \qquad\vdash_S p
\]
where $\vdash_S p$ indicates that $p$ is {\em derivable} in the proof system determined by $S$. After defining validity for atoms one can also define validity for logical connectives via semantic clauses that express proof conditions (e.g. $A \land B$ is provable in $S$ if and only if both $A$ and $B$ are provable in $S$), which results in a framework that evaluates propositions exclusively in terms of proofs of its constituents.

Although the $\Bes$ project has been successfully developed for intuitionistic~\cite{Sandqvist2015IL} and classical logics~\cite{Sandqvist,DBLP:journals/igpl/Makinson14}, it has not yet been systematically developed as a foundation for logical reasoning~\cite{DBLP:journals/synthese/DicherP21,DBLP:journals/jphil/Kurbis15,DBLP:journals/logcom/Francez16a}. In this paper, we intend to move on with this quest, by proposing a $\Bes$ view of {\em ecumenical logics}, inspired by Prawitz's ~\cite{DBLP:journals/Prawitz15} proposal of a system combining classical and intuitionistic logics. 

In Prawitz' system, 
the classical logician and the intuitionistic logician would share the universal quantifier, conjunction, negation and the constant for the absurdity, but they would each have their own existential quantifier, disjunction and implication, with different meanings. 
Prawitz' 
%main idea is that these different meanings are given by a semantical framework that can be accepted by both parties. His 
main motivation was to provide a logical framework that would make possible an inferentialist semantics for the classical logical constants.
%Victor's text
In this way, inferentialism brought forth a very specific proposal when it emerged in the ecumenical context: to provide acceptable assertability conditions for the operators of a certain logical system in another logical system which does not accept them, thus allowing the acceptance and reinterpretation of the previously rejected operators under the new inferential guise. In the context of conflicting discussions between classical and intuitionistic logicians, this would be comparable to defining assertability conditions for classical operators inside intuitionistic logic, which Prawitz actually does in~\cite{DBLP:journals/Prawitz15}. Therefore, the inferentialist's main task is to create ecumenical connectives that, with the assertability conditions exposed in its inferential rules, can represent connectives accepted by one of the logical systems and rejected by the others inside the ecumenical environment. 

%Logical ecumenism thus provides a medium in which meaningful interactions may occur between classical and intuitionistic logic, whilst $\Pts$ provides a way of clarifying what is at stake when one accepts or denies \textit{reductio ad absurdum} as a meaningful proof method.

In this work we {\em do not} intend to provide a $\Bes$ for Prawitz' original system, but rather to proceed with a careful analysis of different aspects of $\Bes$ for logical systems where classical and intuitionistic notions of proof coexist in peace (\ie\ without collapsing). We define intuitionistic proofs through the usual semantic conditions of $\Bes$, which encapsulate the traditional idea of Brouwer, Heyting and Dummett that mathematical existence of an object can only be guaranteed by means of its construction \cite{Brouwer1981-BROBCL,Dummett1977-DUMEOI-2,Heyting1956-HEYIAI-2}. On the other hand, classical proofs are defined by taking into account an idea advanced by David Hilbert to justify non-constructive proof methods: the concept of consistency is conceptually prior to that of truth, and in order to prove the truth of a proposition in a given context it suffices to prove its consistency. In his words \cite{ConsistencyHillbert61554c58-c869-34f0-b322-2cff263d9ae0,Hilbert1979-HILMPL,Hilbert1900}:

%\cyan{Victor: Adicionei esse quote porque acho ele muito bom e tambem porque acho que esse esse e um ponto conceitual muito importante que a gente nao ressaltou tanto em outras partes.}

\begin{quote}
You [Frege] write ``From the truth of the axioms it follows that they do not
contradict one another''. It interested me greatly to read this sentence of yours,
because in fact for as long as i have been thinking, writing and lecturing about
such things, i have always said the very opposite: if arbitrarily chosen axioms
together with everything which follows from them do not contradict one
another, then they are true, and the things defined by the axioms exist. For
me that is the criterion of truth and existence. 
\end{quote}

In order to properly represent this idea of classical proof in $\Bes$ we must change the semantic treatment given to the absurdity constant $\bot$, but it is shown that this can be done without issues.  As expected of an ecumenical framework, the resulting environment allows both notions of proof to coexist peacefully, retain their independence and fruitfully interact  -- so we are able, for instance, to analyze the semantic content of a proposition which is in part proved classically and in part proved intuitionistically in terms of interactions between the respective proof notions.

%In this work we {\em do not} intend to provide a $\Bes$ for Prawitz' original system, but rather to proceed with a careful analysis of different aspects of $\Bes$ for logical systems where classical and intuitionistic notions of proof coexist in peace (\ie\ without collapsing). We define intuitionistic proofs through the usual semantic conditions of $\Bes$, which encapsulate the traditional idea of Brouwer, Heyting and Dummett that mathematical existence of an object can only be guaranteed by means of its construction \cite{Brouwer1981-BROBCL,Dummett1977-DUMEOI-2,Heyting1956-HEYIAI-2}. On the other hand, classical proofs are defined by taking into account an idea advanced by David Hilbert to justify non-constructive proof methods: the concept of consistency is conceptually prior to that of truth, and in order to prove the truth of a proposition in a given context it suffices to prove its consistency \cite{ConsistencyHillbert61554c58-c869-34f0-b322-2cff263d9ae0,Hilbert1979-HILMPL,Hilbert1900}. In order to properly represent this idea of classical proof in $\Bes$ we must change the semantic treatment given to the absurdity constant $\bot$, but it is shown that this can be done without issues.  As expected of an ecumenical framework, the resulting environment allows both notions of proof to coexist peacefully, retain their independence and fruitfully interact  -- so we are able, for instance, to analyze the semantic content of a proposition which is in part proved classically and in part proved intuitionistically in terms of interactions between the respective proof notions.

We start by proposing a {\em weak} version of ecumenical $\Bes$ (in Section~\ref{sec:wes}). This version relies on the concepts of local and global validity (see \eg~\cite{DBLP:journals/sLogica/Cobreros08}), enabling us to examine various aspects of both classical and intuitionistic validities. In particular we demonstrate that, while intuitionistic validity has the property of {\em monotonicity}, meaning that it remains unchanged under extensions, this characteristic does not hold true for classical formulas. This observation gives rise to the motto: 
\begin{quote}
{\em Classical proof plus monotonicity equals intuitionistic proof of double negation.}
\end{quote}
In Section~\ref{sec:ses} we will unwrap the full power of ecumenical $\Bes$, by showing a strong notion of validity. In Section~\ref{sec:ND} 
we define the ecumenical natural deduction system $\NE$, and prove its soundness and completeness  w.r.t. such (proof-theoretic) semantic. We then conclude with some ideas to push forward the $\Pts$ agenda for ecumenical systems.

\section{Base extension semantics}

\subsection{Basic definitions}
We will adopt Sandqvist's~\cite{Sandqvist2015IL} terminology, adapted to the ecumenical setting.
% and Makinson's~\cite{DBLP:journals/igpl/Makinson14} terminology for intuitionistic and classical logics, respectively.

The  propositional {\em base language} is assumed to have a set $\At=\{p_1, p_2,\ldots\}$ of countably many atomic propositions, together with the unit $\bot$.
%Formulas are built from atoms using the binary connectives $\to,\wedge,\vee$ and the unit $\bot$. We use $\neg A$ as an abbreviation for $A \to\bot$, and 
The set $\At\cup\{\bot\}$ will be denoted by $\Atb$, and its elements will be called {\em basic sentences}.

We use, as does Sandqvist, systems containing natural deduction rules over basic sentences for the semantical analysis, and we allow inference rules to discharge sets of basic hypotheses. Unlike Sandqvist, however, we allow the logical constant $\bot$ to be manipulated by the rules.

\begin{definition}[Atomic systems]\label{def:as}
An {\em atomic system} (a.k.a. a  {\em base}) $S$ is a (possibly empty) set of atomic rules of the form
\begin{prooftree}
\AxiomC{$\Gamma_{\At}\;\;[P_1]$}
\noLine
\UnaryInfC{$p_1$}
\AxiomC{$\ldots$}
\AxiomC{$\Gamma_{\At}\;\;[P_n]$}
\noLine
\UnaryInfC{$p_n$}
\TrinaryInfC{$p$}
\end{prooftree}
where $p_i,p\in\Atb$ and $\Gamma_{\At},P_i$ are (possibly empty) sets of basic sentences. The sequence $\langle  p_{1}, \ldots , p_{n} \rangle$ of premises of the rule can be empty -- in this case the rule is called an {\em atomic axiom}. 
%\magenta{EP. I've added the gammas to be coherent with the rest of the text. I guess we do not have to add them to the sequence above? Btw, do we use this notation anywhere else?} \cyan{Nao usamos essa notacao em nenhum outro lugar, inclusive eu so adicionei ela ontem (pra deixar o comentario das "empty premises" um pouco mais claro!). Mas dado que vamos omitir um dos conjuntos eu acho melhor omitirmos os dois; acabei de mudar.}
\end{definition}

Labels will sometimes be written as the superscript of $[P_{i}]$ and to the right of a rule to denote that $P_{i}$ was discharged at that rule application. 

\begin{definition}[Extensions]\label{def:exts}
An atomic system $S'$ is an {\em extension} of an atomic system $S$ (written $S \subseteq S'$), if $S'$ results from adding a (possibly empty) set of atomic rules to $S$.   
\end{definition}

%If $S$ is an atomic system and $K$ is a set of sentence letters, we write $\overline{S(K)}$ for the {\em closure} of $K$ under all the rules in $S$ in the usual set-theoretic sense. An atomic system $S$ is {\em dense} if $\overline{S(\emptyset)}=\At$.

\begin{definition}[Deducibility]\label{def:deduc}
  The deducibility relation $\vdash_S$ coincides with the usual notion in the system of natural deduction consisting of just the rules in $S$, that is, $p_1,\ldots,p_n \vdash_S p$ iff there exists a deduction with the rules of $S$ whose conclusion is $p$ and whose set of undischarged premises is a subset of $\{p_{1}, ... , p_{n}\}$.
\end{definition}

\begin{definition}[Consistency]\label{def:deduc}
An atomic system $S$ is called {\em consistent} if $\nvdash_S \bot$.
\end{definition}

Notice that, due to how deducibility is defined, if  $p \in \{p_{1}, ... , p_{n}\}$ then $p_1, ..., p_{n} \vdash_{S} p$ holds regardless of the $S$, as a single occurrence of the assumption $p$ already counts as a ruleless deduction of $p$ from $p$.

Those basic definitions are usually combined with validity clauses to obtain semantics for intuitionistic logic. For instance, Sandqvist \cite{Sandqvist2015IL} defines atomic derivability using $\At$ instead of $\Atb$ and employs the following clauses:

\begin{definition}[Base-extension semantics]\label{def:wvalidity} Validity in $\Bes$ is defined as follows:
    
\begin{enumerate}
    \item $\Vdash_{S} p$ iff $\vdash_S a$, for  $p \in\At$;

\medskip

    \item $\Vdash_{S} (A \land B)$ iff $\Vdash_{S} A$ and $\Vdash_{S} B$;

\medskip      

     \item $\Vdash_{S} (A \to B)$ iff $A \Vdash_{S}  B$; 

\medskip       

      \item $\Vdash_{S} A \lor B$ iff  $ \forall S' (S \subseteq S') $ and all $p \in \At$, $A \Vdash_{S'} p$ and $B \Vdash_{S'} p$ implies $\Vdash_{S'}  p$;

\medskip

\item $\Vdash_{S} \bot$ iff $\Vdash_{S} p$ for all $p \in \At$

\medskip

\item For non-empty $\Gamma$, $\Gamma \Vdash_{S} A$ iff for all $S'$ such that $S \subseteq S'$ it holds that, if $\Vdash_{S'} B$ for all $B \in \Gamma$, then $\Vdash_{S'} A$;

\medskip

\item $\Gamma \Vdash_{\Bes} A$ iff $\Gamma \Vdash_{S} A$ for all $S$;

\end{enumerate}
\end{definition}

The idea being that $\Bes$ validity ($\Vdash_{\Bes}$) is defined in terms of $S$-validity ($\Vdash_{S}$) and $S$-validity is reducible to derivability in $S$ and its extensions, so we obtain a semantics defined exclusively in terms of {\em proofs} and {\em proof conditions}. In this sense, $\Bes$ not only aims at elucidating the meaning of a logical proof, but also providing means for its use as a basic concept of semantic analysis.

\subsection{On the semantics of $\bot$ in $\Bes$}\label{sec:bot}

The semantic conditions for $\bot$ are usually defined in $\Bes$ in one of two ways. The first one is to define atomic derivability by using $\At$ instead of $\Atb$ and employ the following semantic clause: 
%\cite{Sandqvist2015IL}:
\[
\Vdash_{S} \bot \qquad \mbox{iff} \qquad \Vdash_{S} p \ \mbox{holds for all} \ p \in \At
\]
Absurdity is treated as a logical constant and cannot figure in atomic bases, hence the switch from $\Atb$ to $\At$. This clause, used most notably by Sandqvist \cite{Sandqvist2015IL}, borrows from Dummett \cite{dummett1991logical} the idea of defining absurdity in terms of logical explosion, but restricts it to just atoms in order to make the definition inductive. 

The second one is to consider $\bot$ an atom and require all bases to contain atomic \textit{ex falso} rules concluding $p$ from $\bot$ for every $p \in \Atb$ \cite{piecha2015failure}. If for some $S$ we have $\Vdash_{S} \bot$ this now implies $\vdash_{S} \bot$; hence, for any $p \in \Atb$, the deduction of $\bot$ from empty premises in $S$ can be extended by the appropriate \textit{ex falso} rule to a deduction showing $\vdash_{S} p$ that also shows $\Vdash_{S} p$. Since $\Vdash_{S} p$ for all $p \in \Atb$ also implies $\Vdash_{S} \bot$ this means that ($\Vdash_{S} \bot$ iff $\Vdash_{S} p$ for all $p \in \Atb$) holds in all such $S$, so by restricting ourselves to these ``atomically explosive" bases we end up giving the same semantic treatment to $\bot$.

Although technically sound, such definitions are inadequate from a conceptual standpoint because they do not express the intended meaning of $\bot$. The absurdity constant is supposed to represent a statement that never holds. Logical explosion should be a \textit{consequence} of the definition of absurdity, not its definition. Ideally, $\bot \Vdash_{S} A$ should hold for all $S$ because no extension $S'$ of $S$ validates $\bot$ and the entailment holds vacuously, not because there are extensions validating $\bot$ that also validate the arbitrary formula $A$. 

Unfortunately, we cannot define $\bot$ through the natural semantic clause:
\[
\nVdash_{S} \bot \qquad \mbox{for all $S$.}
\]
If we define $\neg A$ as $(A \to \bot)$ and adopt the clause above, we can demonstrate that $\Vdash_{S} \neg \neg p$ holds for every $p \in \At$ in every $S$. This follows because $\Vdash_{S} \neg p$ would hold only if no extension of $S$ validates $p$. However, for every $S$ and every $p \in \At$, there is always some extension $S'$ of $S$ such that $\Vdash_{S'} p$ (for instance, the extension obtained by adding the atomic axiom with conclusion $p$ to $S$). Piecha et al.~\cite{piecha2015failure} observe that, aside from ruling out this definition of $\bot$, the fact that every atom is validated in some extension of every system ``\textit{might be considered a fault of validity-based proof-theoretic semantics, since it speaks against the intuitionistic idea of negation $\neg A$ as expressing that $A$ can never be verified}".

As will be shown in this paper, a technically sound and conceptually adequate treatment of $\bot$ has been thus far overlooked by the literature. Even though we cannot define $\bot$ in terms of unsatisfiability through a semantic clause, it is still possible to do it by simply requiring all bases to be consistent:
\[
\nvdash_{S} \bot \qquad \mbox{for all $S$.}
\]
While at first glance it may seem that the definitions are equivalent, this switch actually allows us to solve the issues with the semantic clause. Moreover, the restriction implements the desiderata of Piecha et al. and allows bases to contain no extensions validating some specific atoms. To see why, consider that if $S$ is a consistent base in which $p \vdash_{S} \bot$ holds then there can be no extension $S'$ of $S$ validating $p$, since if $\Vdash_{S'} p$ for some $S \subseteq S'$ we would have a deduction showing $\vdash_{S'} p$ which could be composed with the deduction showing $p \vdash_{S} \bot$ (which is also a deduction showing $p \vdash_{S'} \bot$, since all rules of $S$ are in $S'$ by the definition of extension) to obtain one showing $\vdash_{S'} \bot$, hence $S'$ would be inconsistent. It is easy to show that $p \vdash_{S} \bot$ now implies both $\Vdash_{S} \neg p$ and $\nVdash_{S} \neg  \neg p$. As such, by considering $\bot$ an atom but requiring it to always be underivable we allow bases to indirectly restrict their own admissible extensions by conveying information about which formulas will never be validates in their extensions.

The completeness proof presented in Section \ref{sec:ND} can easily be adapted to standard intuitionistic $\Bes$ by simply omitting all steps concerning classical formulas. Since the remaining steps are precisely the constructive ones, this yields a fully constructive proof of completeness for $\Bes$ with the consistency constraint and without any semantic clause for $\bot$ w.r.t. intuitionistic propositional logic. It should also be noted that many of the results we will demonstrate for ecumenical logics are essentially dependent on the consistency constraint, so the definition is optional in intuitionistic logic but essential in our ecumenical semantics.

%If $\{a,a_1,\ldots,a_n\}\in\Atb$, the {\em derivability} of $a$ from the (possibly empty) set $\{a_1,\ldots,a_n\}$ of assumptions in an atomic system $S$ is written $a_1,\ldots,a_n \vdash_S a$. 
%That is, $a$ is derivable from $\{a_1,\ldots,a_n\}$ iff $a\in\overline{S(a_1,\ldots,a_n)}$. \red{EP. It may be the case that we drop the notion of closure...}

%Observe that $\overline{S(K)}$ is the intersection of all sets $L$ of sentence letters such that $K \subseteq L$ and whenever a rule with premises $a_1,\ldots,a_n$ and conclusion $a$ is in $S$ with  $\{a_1,\ldots,a_n\}\subset L$ then $a\in L$.

\subsection{Ecumenical language}

%In the ecumenical setting we require all atomic systems to be consistent.
%, that is, $\bot\notin\overline{S(\emptyset)}$. 
%In other words,for all $S$ it must be the case that $\nvdash_S \bot$. 
%Hence, whenever considering atomic systems $S$ and extensions of $S$, it is implicitly considered that only consistent extensions are being taken into account.

%\cyan{VN: Vou fazer mudanas a partir daqui! A primeira e mais import‡nte Ž que agora todas as f—rmulas tm uma vers‹o cl‡ssica, e eu tirei o superscript de f—rmulas intuicionistas.}

Propositional formulas are built from basic sentences using the binary connectives $\to,\wedge,\vee$. The ecumenical language is defined as follows.

\begin{definition}
    The ecumenical language $\mathcal{L}$ is comprised of the following ecumenical formulas.

    %all classical and intuitionstic formulas, defined as follows.
    \begin{enumerate}
        \item If $p \in \Atb$, them $p^i,p^c\in \el$;
        %are formulas;

\medskip
        
        \item If $A,B\in\el$, then $(A \land B)^i$, $(A \lor B)^i,(A \to B)^i\in\el$;% are formulas;
        %\item If $A,B\in\el$, then $A \land B$, $A \lor B,A \to B\in\el$;% are formulas;

\medskip
        
        \item If $A,B\in \el$, then $(A \land B)^c$, $(A \lor B)^c,(A \to B)^c\in\el$;% are formulas;
       % \item Nothing else is a formula.
    \end{enumerate}

\end{definition}

%and then extend it by including a formula $A^c$ for every formula $A$ in the original language\footnote{Reiteration of this clause is not allowed, and so $(A^c)^c$ is not a formula of our language.}. We call those formulas \textit{classical}, and also define that every formula which is not classical is \textit{intuitionistic}. 

%Those definitions avoid iterations of the classical superscript and prevent formulas such as $(A^c)^c$ from belonging our language. Notice that, for $\Box \in \{ \to, \land, \lor \}$, $A^c \Box B^c$, $A^c \Box B$, $A \Box B^c$, $A \Box B$, $(A^c \Box B^c)^c$, $(A^c \Box B)^c$, $(A \Box B^c)^c$, $(A \Box B)^c$ are all formulas of our language, as $p$ and $p^c$ for any atom $p$, $\bot$, and $\bot^c$.

\begin{notation}
%Superscripts will be dropped whenever they are irrelevant,
%We write simply $A$ whenever the supescript of a formula is irrelevant - that is, in contexts in which the formulas $A^i$ and $A^c$ receive equal treatment. 
%and 

Parenthesis are omitted whenever no confusion ensues. 
%When convenient, $A \star^i B$ and $A \star^c B$  will also be used as abbreviations of $(A \star B)^i$ and $(A \star B)^c$  for $\star \in \{ \land, \lor, \to \}$. 
For easing the notation,
%we will represent $\bot$ by simply $\bot$ and 
$\neg A$, $A \to B$, $A \land B$ and $A \lor B$ will be abbreviations of $(A \to \bot)^i$, $(A \to B)^i$  $(A \land B)^i$,  $(A \lor B)^i$, respectively. Finally, we stipulate that if a formula $A$ is used without specification of its superscript, then it may be either $i$ or $c$. For instance, $A^{i} \land B^{c}$ should be read as a placeholder for $(A^{i} \land B^{c})^{i}$, but $A \land B$ should be read as a placeholder for $(A^{i} \land B^{i})^{i}$,   $(A^{c} \land B^{i})^{i}$,  $(A^{i} \land B^{c})^{i}$ and  $(A^{c} \land B^{c})^{i}$. 
%\magenta{EP. This is a bit confusing, since we have just suppressed the superscripts with the parenthesis, no?} \cyan{Verdade, mas acho que vai demorar um bocado pra adaptar tudo. Entao acho melhor exemplificar pra esclarecer!}
%\cyan{Adicionei texto extra (comecando no "For instance") pra exemplificar.}

\end{notation}

%\cyan{VN: Mudei a nota‹o!} \red{EP. Negations also collapse?} \cyan{VN: Colapsam sim, e em todos os resultados n—s s— usamos a intuicionista!}

%Later on we will also drop the superscrits for $\bot$, and $\bot$ will be used as an abbreviation of $\bot$ in semantic contexts.
%Finally, $\neg A$ will be an abbreviation of $((A \to (\bot))^i$. Later on we will also drop the superscrits for $\bot$, and $\bot$ will be used as an abbreviation of $\bot$ in semantic contexts.

\begin{definition}\label{def:complexity}
The {\em complexity} of a formula with shape $A^{i}$ is the number of logical operators distinct from $\bot$ occurring on it. The {\em complexity} of a formula with shape $A^{c}$ is the complexity of $A^{i}$ plus 1. 
\end{definition}

Intuitively, a \textit{intuitionistic formula} $A^i$ holds whenever there exists an intuitionistic proof of $A$, and a \textit{classical formula} $A^c$ holds whenever there exists a classical proof of $A$. Since every formula of the usual language has both a classical and an intuitionistic version, every formula (except the absurdity constant) is allowed to have classical and intuitionistic proofs.

%Our study 
In this paper, we will focus on two definitions of semantic ecumenism,  called \textit{weak} and \textit{strong} ecumenical semantics, respectively. In both the semantics of classical proofs is given in terms of the consistency of formulas w.r.t. some atomic system, 
%but the semantic clauses for classical formulas of the weak ecumenical semantics are much weaker than that of the strong ecumenical semantics. 
%
%Remarkably, both 
but the notions induce the classical behaviour in very different ways. 

It should be observed that the weak ecumenical semantics proposed next does not have a simple syntactic characterization, and its study is meant for semantic purposes only -- the goal is to explore deeply the ecumenical proof-theoretic behaviour. 
We present, in Section~\ref{sec:ND}, an interesting ecumenical natural deduction system which is sound and complete w.r.t. 
the strong ecumenical semantics described in Section~\ref{sec:ses}.

%Atoms occurring on atomic systems themselves must always be unlabelled, but atoms occurring in formulas must always be labelled. $p^i$ codifies the assertability that there is a intuitionistic proof of $p$, and $p^c$ codifies the assertability that there is a classical proof of $p$. 

%Only atoms receive labels, and thus whenever labels are used it is implicitly stated that we are dealing with atomic formulas. We will call {\em classical} formulas containing only classical labelled atoms, {\em intuitionistic} formulas containing only intuitionistic labelled atoms and {\em ecumenical} will be general formulas.

\section{Weak ecumenical semantics}\label{sec:wes}

We start by distinguishing between two notions of logical validity for every atomic system $S$: \textit{local} logical validity (represented by $\Vdash^{L}_{S}$) and \textit{global} logical validity (represented by $\Vdash^{G}_{S}$). The idea is that, in local validity, we only concern ourselves with what holds in a base and its extensions, but in global validity we also take into account what holds in extensions of a base's extensions. In many contexts both notions collapse, but in some there is reason to make a distinction. Some relationships between both will be studied in more depth later.

A weak ecumenical version of $\Bes$ is given by the definition below. As usual, we start by giving semantic conditions for basic sentences in atomic systems,  expanding them through semantic clauses.

\begin{definition}[Weak Validity]\label{def:wvalidity} Weak $S$-validity and weak validity are defined as follows.
    
\begin{enumerate}
    \item $\Vdash_{S}^{L} p^i$ iff $\vdash_S p$, for  $p \in\Atb$;

    \medskip    
    
    \item $\Vdash_{S}^{L} p^c$ iff $p \nvdash_S \bot$, for  $p \in \Atb$;

\medskip    
    
    \item $\Vdash_{S}^{L} A^c$ iff $A^i \nVdash^{L}_S \bot^{i}$, for $A \notin \Atb$;

 \medskip     
        
    \item $\Vdash_{S}^{L} (A \land B)^i$ iff $\Vdash_{S}^{L} A$ and $\Vdash_{S}^{L} B$;

\medskip      

     \item $\Vdash_{S}^{L} (A \to B)^i$ iff $A \Vdash_{S}^{G}  B$; 

\medskip       

      \item $\Vdash_{S}^{L} (A \lor B)^i$ iff  $ \forall S' (S \subseteq S') $ and all $p \in \Atb$, $A \Vdash_{S'}^{L} p^i$ and $B \Vdash_{S'}^{L} p^i$ implies $\Vdash_{S'}^{L}  p^i$;
      
\medskip  

\item For non-empty $\Gamma$, we have that $\Gamma \Vdash_{S}^{L} A$ iff for all $S'$ such that $S \subseteq S'$ it holds that, if $\Vdash_{S'}^{L} B$ for all $B \in \Gamma$, then $\Vdash_{S'}^{L} A$;

\medskip

\item $\Gamma \Vdash_{S}^{G}  A$ iff for all $S'$ such that $S \subseteq S'$ we have that, if for all $S''$ such that $S' \subseteq S''$ it holds that $\Vdash_{S''}^{L} B$ for all $B \in \Gamma$, then for all $S''$ such that $S' \subseteq S''$ it also holds that $\Vdash_{S''}^{L} A$;

\medskip

%\medskip

%\item $\Vdash_{S}^{G} \Gamma$ iff, for all $S \subseteq S^n$, $\Vdash_{S^{n}}^{L} B$, for all $B \in \Gamma$.

%\medskip

%\item For non-empty $\Gamma$, we have that $\Gamma \Vdash_{S}^G A$ iff, for all $S \subseteq S^n$, $\Vdash_{S}^{G} \Gamma$ implies $\Vdash_{S}^{G} A$.

%\Gamma \Vdash_{S}^{G}  A$ iff for all $S^n$ such that $S \subseteq S^n$ we have that, if for all $S^m$ such that $S^n \subseteq S^m$ it holds that $\Vdash_{S^{m}}^{L} B$ for all $B \in \Gamma$, then for all $S^m$ such that $S^n \subseteq S^m$it also holds that $\Vdash_{S^{m}}^{L} A$;

      \item $\Gamma \Vdash A $ iff $\Gamma \Vdash_{S}^{G} A$ for all $S$.
\end{enumerate}
\end{definition}

%CONTINUAR DAQUI It's important to take note of some choices made in the definition of this semantics. First and foremost, our consistency requirement gives a distinct treatment to $\bot$ when compared to the semantic clause used by Sanqdvist\footnote{In~\cite{Sandqvist2015IL} the $S$-validity of $\bot$ is defined as: $\Vdash_S\bot$  iff $\Vdash_S p$ for every atomic $p$.}. The definition of disjunction takes into account local validity instead of global; by defining it in terms of global entailment we would approach the strong ecumenical semantics, which will only be presented later.

There are important bits of information to unpack in those clauses. First, notice that there is one clause for classical proofs of atoms and one for classical proofs of non-atomic formulas, but both are defined in terms of consistency proofs for the formula's immediate subformula. Second, while Clause 7 is the same as Sandqvist's clause ``(Inf)", Clause 8  is slightly more complex; the former is our definition of \textit{local} validity, the later of \textit{global} validity ~\cite{DBLP:journals/sLogica/Cobreros08}. This distinction is redundant in usual intuitionistic semantics, but essential in the weak version of ecumenical $\Bes$. Finally, notice that when defining the semantic clause for disjunction we use local entailment instead of global, which is done to show that some desirable semantic properties follows from this weak definition. By using the global notion instead we would obtain an alternative presentation of what we later define as strong ecumenical semantics.

%Clause 8 could also be rewritten as follows (and we shall use both definitions interchangeably).
%\begin{lemma}\label{lemma:equivalentdefinitionsofG}
 %   $\Gamma \Vdash^{G}_ S A$ iff ($\Vdash_{S}^{G} \Gamma$ implies $\Vdash^{G}_{S} A$).
%\end{lemma}

%\begin{proof}
 %   ($\Rightarrow$) Assume $\Gamma \Vdash^{G}_ S A$. Then for every extension $S \subseteq S^n$ the following holds: if for every $S^n \subseteq S^m$ we have that $\Vdash_{S^m}^{L} B$ holds for every $B \in \Gamma$, then for every $S^n \subseteq S^m$ it also holds that $\Vdash_{S^m}^{L} A$. Now assume $\Vdash_{S}^{G} \Gamma$. By Clause 7 we have that, for every $S \subseteq S^j$, it holds for every $S^j \subseteq S^k$ that $\Vdash^{L}_{S^k} B$ for every $B \in \Gamma$. By putting $S^j = S^k = S$ we have $\Vdash^{L}_{S} B$ for every $B \in \Gamma$. And then we put $S^n = S$ to conclude $\Vdash_{S^m}^{L} A$ for all extension $S^m$ of $S$, which yields  $\Vdash^{G}_{S} A$.

 %   ($\Leftarrow$) Assume ($\Vdash_{S}^{G} \Gamma$ implies $\Vdash^{G}_{S} A$). Then the following holds: if for all $S \subseteq S^n$ we have that for all $S^m$ such that $S^n \subseteq S^m$ it holds that $\Vdash_{S^{m}}^{L} B$ for all $B \in \Gamma$ (which is the same as $\Vdash_{S}^{G} \Gamma$), then for all $S^m$ such that $S^n \subseteq S^m$ it also holds that $\Vdash_{S^{m}}^{L} A$ (which is the same as $\Vdash^{G}_{S} A$). But this is just a restatement of Clause 7, which proves the result.  
%\end{proof}

The following result easily follow from Definition~\ref{def:wvalidity} and the requirement of atomic systems to be consistent.

\begin{lemma}\label{lemma:bot}
$\nVdash^{L}_{S} \bot^i$ and  $\nVdash^{L}_{S} \bot^c$ for all $S$.
\end{lemma}

\begin{proof}
For any $S$, since $S$ must be consistent we have $\nvdash_{S} \bot$, so $\nVdash^{L}_{S} \bot^i$. Moreover, since $\bot \vdash_{S} \bot$ holds by the definition of deducibility we have  $\nVdash^{L}_{S} \bot^c$.
\end{proof}

%\cyan{VN: Acho que faltaram os superscripts nos bottoms, n‹o? Coloquei eles, se estiver errado Ž s— tirar!}

Due to this result, from now on $\bot$ will be used as an abbreviation of $\bot^i$ ($\equiv\bot^c$) in semantic contexts.

%    For $\bot$ this follows immediately from Clause 1 of Definition \ref{def:wvalidity} and the consistency requirement; for $\bot^c$ this follows immediately from the fact that $\bot \vdash_S \bot$ holds for any $S$ by the definition of atomic deduction (see Definition~\ref{def:as}).  

%\begin{notation}
%Since only one instance of $\bot$ is required in our results, from now on we write simply $\bot$ whenever we wish to refer specifically to $\bot$.
%\end{notation}

Although expected, since intuitionistic provability implies classical provability, the next two results are only possible due to
the change from Sandqvist's clause for $\bot$ to the consistency requirement.

\begin{theorem}\label{lemma:intimpliesclasatom}
$p^{i} \Vdash_{S}^{L} p^{c}$ for any $p \in \Atb$.
     %$p^i \Vdash_S^{L} p^c$.  %\red{I think that the statement is stronger, right? $p^i\Vdash p^c$...} \cyan{VN: That's right! I've corrected it!}
\end{theorem}
\begin{proof} 
%{\em Atomic case.}
%Assume that, for some $S \subseteq S^n$, $\Vdash_{S^n}^{L} a^i$. By Clause 1, $\vdash_{S^n} a$. If $a \vdash_{S^n} \bot$ them, by composition, $\vdash_{S^n} \bot$. But this is impossible, given that every atomic system $S^n$ is consistent. Thus, $a \nvdash_{S^n} \bot$ and, by Clause 2, $\Vdash^{L}_{S^n} a^c$. Since $S^n$ is an arbitrary extension of $S$, we have $a^i \Vdash_{S}^{L} a^c$ by Clause 7. 
% 
%\end{proof}
%
%
%\begin{theorem}\label{lemma:intimpliesclas2}
%     $A^i \Vdash_S^{L} A^c$, for $A \notin \Atb$.
%\end{theorem}
%\begin{proof} 
%{\em General case.} 

Assume $\Vdash_{S'}^{L} p^{i}$ for some $S \subseteq S'$. Then $\vdash_{S'} p$. Now suppose $p \vdash_{S'} \bot$. Then by composing both deductions we have $\vdash_{S'} \bot$, contradicting the consistency requirement. So $p \nvdash_{S'} \bot$, hence $\Vdash^{L}_{S'} p^{c}$. Since $S'$ is an arbitrary extension of $S$, we have $p^i \Vdash_{S}^{L} p^c$ by Clause 7.  
\end{proof}

\begin{theorem}\label{lemma:intimpliesclas}
$A^i \Vdash_S^{L} A^c$ for any $A \notin \Atb$.
     %$p^i \Vdash_S^{L} p^c$.  %\red{I think that the statement is stronger, right? $p^i\Vdash p^c$...} \cyan{VN: That's right! I've corrected it!}
\end{theorem}
\begin{proof}

Assume  $\Vdash_{S'}^{L} A^i$ for some $S \subseteq S'$ and suppose that $A^i \Vdash_{S'}^{L} \bot$. Then by Clause 7 of Definition~\ref{def:wvalidity} we have $\Vdash_{S'}^{L} \bot$, and then $\vdash_{S'} \bot$ by Clause 1, which is a contradiction. Thus, $A^i \nVdash_{S'}^{L} \bot$, and so $\Vdash_{S'}^{L} A^c$. Since $S'$ is an arbitrary extension of $S$, we have $A^i \Vdash_{S}^{L} A^c$ by Clause 7.  
\end{proof}

If Sandqvist's definition was used, from  $\Vdash_{S'}^{L} A^i$ and  $A^i \Vdash_{S'}^{L} \bot$ we could get $\Vdash_{S'}^{L} p^i$ for arbitrary $p \in \Atb$, but it would not be the case that $A^i \nVdash_{S'}^{L} \bot$. The same would happen with the proof for atoms if we allowed $\bot$ to occur in atomic bases and required all bases to contain all instances of the atomic \textit{ex falso}.

\subsection{Monotonicity}
It is well known that $\Bes$ validity in intuitionistic logic is {\em monotonic}, in the sense that it is stable under base extensions. As it turns out, this is not the case in the ecumenical setting, as discussed next.
%We define monotonicity for formulas as follows.
\begin{definition}[Monotonicity] A formula $A$ is called {\em $S$-monotonic} with respect to an atomic system $S$ if, for all $S\subseteq S'$,  $\Vdash_{S}^{L} A$ implies $\Vdash_{S'}^{L} A$. $A$ is called {\em monotonic} if it is $S$-monotonic for any atomic system $S$.
\end{definition}

Some parts of Clause 8 come for free in the presence of monotonicity (as shown next), 
but they must be explicitly stated on the lack of it. As such, the original notion of logical consequence provides only a weak kind of validation for non-monotonic formulas, and thus would indirectly treat classical and intuitionistic formulas {\em very differently}.

\begin{theorem}\label{thm:monotoniccollapse}
     If $S$-monotonicity holds for $A$ and all formulas in $\Gamma$, then $\Gamma \Vdash_{S}^{L} A$ iff $ \Gamma \Vdash_{S}^{G} A$.
\end{theorem}

\begin{proof}
The result is trivial if $\Gamma=\emptyset$, so we will assume $\Gamma$ non-empty.

\medskip

($\Rightarrow$) Suppose $\Gamma \Vdash_{S}^{L} A$. Then, by Clause 7, for every $S'$ such that $S \subseteq S'$ we have that, if $\Vdash_{S'}^{L} B$ for every $B \in \Gamma$ then $\Vdash_{S'}^{L} A$. Now pick any such $S'$ in which, for all $S''$ such that $S' \subseteq S''$, we have $\Vdash_{S''}^{L} B$ for all $B \in \Gamma$. Since $S' \subseteq S'$, we have $\Vdash_{S'}^{L} B$ for every $B \in \Gamma$, so we conclude $\Vdash_{S'}^{L} A$. But by $S$-monotonicity we also have that $\Vdash_{S'}^{L} A$ implies $\Vdash_{S''}^{L} A$ for all $S''$ such that $S' \subseteq S''$ and, since $S'$ was an arbitrary extension of $S$ satisfying the antecedent of the second part of Clause 7, we have $\Gamma \Vdash_{S}^{G} A$.

\medskip

($\Leftarrow$) Assume $\Gamma \Vdash_{S}^{G} A$. Then, by Clause 7, for every $S'$ such that $S \subseteq S'$ we have that, if $\Vdash_{S''}^{L} B$ for every $B \in \Gamma$ and for all $S''$ such that $S' \subseteq S''$ then $\Vdash_{S''}^{L} A$ for all such $S''$ as well. Now let $S'$ be any extension of $S$ such that $ \Vdash_{S'}^{L} B$ for all $B \in \Gamma$. By monotonicity, for all formulas $B \in \Gamma$ we have that, if $ \Vdash_{S'}^{L} B$, then $ \Vdash_{S''}^{L} B$ for every $S''$ such that $S' \subseteq S''$. Taken together with our assumption, this yields $ \Vdash_{S''}^{L} A$ for all such $S''$. In particular, since $S' \subseteq S'$ we have $ \Vdash_{S'}^{L} A$ and, since $S'$ was an arbitrary extension of $S$ satisfying the antecedent of the second part of Clause 7, we have $\Gamma \Vdash_{S}^{L} A$.  
\end{proof}

%Since to evaluate validity and consequence relations for intuitionistic formulas 
%which do not contain classical atoms among its subformulas 
%we only use traditional clauses for intuitionistic logic (which are known to yield monotonicity for arbitrary formulas - see~\cite[Lemma 3.2. (a)]{Sandqvist2015IL}), the following result is easily proven by induction on the complexity of formulas.

%\begin{proposition}
%Let $A$ be a formula not containing any classical atom as a subformula. Then $A$ is monotone.
%\end{proposition}

Even though intuitionistic atoms and connectives are monotonic, this is not the case in the classical setting. 

%Indeed, for some $S$ (such as the atomic system containing no rules), we have that $p \nvdash_S \bot$ but also have that $S$ can be extended to a consistent system $S_n$ in which $p \vdash_{S^n} \bot$. 
\begin{theorem}\label{thm:monotonicity}
    Every formula  containing only intuitionistic subformulas is monotonic. Classical atoms are not monotonic.
\end{theorem}
\begin{proof}
The result for formulas containing only intuitionistic subformulas is easily proven by induction on the complexity of formulas in the same way as in~\cite[Lemma 3.2. (a)]{Sandqvist2015IL}, where the induction hypothesis is only needed for conjunction (the case for implication holds directly from the definition of general validity).

Regarding classical atoms, for $S=\emptyset$ we have that $p \nvdash_S \bot$ for every $p\in\At$. But if $S'$ is the atomic system containing only the rule obtaining $\bot$ from $p$, $S'$ is consistent and $p \vdash_{S'} \bot$.
Hence $\Vdash_{S}^{L} p^c$ and $S \subseteq S'$,  but $\nVdash_{S'}^{L} p^c$. More generally, if $p$ does not occur in the rules of $S$ then $p \nvdash_{S} \bot$ and by adding a rule obtaining $\bot$ from $p$ to $S$ we have an extension $S'$ guaranteed to be consistent, so whenever $p$ does not appear on the rules of $S$ we have $\Vdash_{S}^{L} p^c$ but $\nVdash_{S'}^{L} p^c$ for some $S \subseteq S'$.

%For a more general counterexample, notice that if an atom $p$ does not occur in a system $S$ then $p \nvdash_{S} \bot$, and since $p$ does not occur in the rules of $S$ by adding a rule obtaining $\bot$ from $p$ to $S$ we obtain a extension $S^n$ guaranteed to be consistent, and the above reasoning can once again be used to show $\Vdash_{S}^{L} p^c$ but $\nVdash_{S^n}^{L} p^c$.     
\end{proof}
%As such, 
In short, for intuitionistic formulas it is irrelevant whether local or global notions of validity is used. For ecumenical formulas containing classical subformulas, however, this choice makes an enormous difference, as illustrated in Section~\ref{sec:eb}.

\subsection{Basic lemmata}

Before proceeding, we shall briefly present some lemmas which will be handy in what comes next.

We start by showing that local validity implies global validity only for non-empty contexts, but global validity implies local validity only when the context is empty. A counter-example for Lemma~\ref{lemma:globaltheoremimplieslocaltheorem} with non-empty contexts is given in Theorem~\ref{thm:Acai}, whereas one for Lemma \ref{lemma:localimpliesglobal} with empty contexts can be obtained by putting $A = p^{c}$ and remembering that classical atoms are not always monotonic.

\begin{lemma}\label{lemma:localimpliesglobal}
 For non-empty $\Gamma$,   $\Gamma \Vdash^L_S A$ implies $\Gamma \Vdash^G_S A$.
\end{lemma}

\begin{proof}
    Let $\Gamma \Vdash^L_S A$. Then, by the definition of local consequence, for every $S'$ such that $S \subseteq S'$ we have that if $\Vdash^L_{S'} B$ for all $B \in \Gamma$ then $\Vdash^L_{S'} A$. Now, let $S'$ be any extension of $S$ in which for all $S''$ such that $S' \subseteq S''$, we have $\Vdash^L_{S''} B$ for all $B \in \Gamma$. Consider any such $S''$. Since $\Gamma \Vdash^L_{S} A$ holds and $S''$ is also an extension of $S$ (by transitivity of extension), we have that $\Vdash^L_{S''} B$ for all $B \in \Gamma$ implies $\Vdash^L_{S''} A$. Since $S''$ is an extension of $S'$, by definition we have $\Vdash^L_{S''} B$ for all $B \in \Gamma$, and thus we have $\Vdash^L_{S''} A$. But this holds for arbitrary $S''$ extending $S'$, and so for every $S''$ we have $\Vdash^L_{S''} A$. Since $S'$ is an arbitrary extension of $S$ satisfying the antecedent of Clause 8 of the Definition \ref{def:wvalidity}, we conclude $\Gamma \Vdash^G_S A$. 
    
\end{proof}

%%%DOWN4

\begin{lemma}\label{lemma:globaltheoremimplieslocaltheorem}
    $\Vdash_{S}^{G} A$ implies $\Vdash_{S}^{L} A$. %\red{EP. I think it is $\Vdash_{S}^{G} A$ implies $\Vdash_{S}^{L} A$} \cyan{VN: Corrigi!}
\end{lemma}

\begin{proof}
    Assume $\Vdash_{S}^{G} A$. Then for every $S \subseteq S'$, it holds  that $\Vdash^{L}_{S''} A$, for every $S' \subseteq S''$. By putting $S'' = S' = S$ we conclude $\Vdash^{L}_{S} A$.  

\end{proof}

\begin{lemma}\label{lemma:newsimplification}
  If  $\nVdash_{S'}^{L} A$ for all $S \subseteq S'$ then $\Vdash_{S'}^{L} \neg A$ for all $S \subseteq S'$.
\end{lemma}

\begin{proof}
    Assume $\nVdash_{S'}^{L} A$ for all $S \subseteq S'$. Pick any such extension $S'$. For any $S' \subseteq S''$ by transitivity of the extension relation we have $S \subseteq S''$, so $\nVdash_{S''}^{L} A$. But then clearly $A^{i} \Vdash^{G}_{S'} \bot$ is satisfied vacuously for any $S'$, so  $\Vdash_{S'}^{L} \neg A$ for all $S \subseteq S'$.     

\end{proof}

The following is a form of global {\em modus ponens}. 
\begin{lemma}\label{lemma:globalmodusponens}
    $\Vdash_{S}^{G} A$ and $A \Vdash_{S}^{G} B$ implies $\Vdash_{S}^{G} B$.
\end{lemma}

\begin{proof}
Assume $\Vdash_{S}^{G} A$. Thus for all $S \subseteq S'$ we have that $S' \subseteq S''$ implies $\Vdash_{S''}^{L} A$. Assume $A \Vdash_{S}^{G} B$. Then, for any $S \subseteq S'$, if for all $S' \subseteq S''$ we have $\Vdash_{S''}^{L} A$, then for all $S' \subseteq S''$ we have $\Vdash_{S''}^{L} B$. By putting $S = S'$ the antecedent gets satisfied we immediately get $\Vdash_{S''}^{L} B$ for all $S \subseteq S''$, hence  $\Vdash_{S}^{G} B$.   
\end{proof}

Finally, the following results show interactions between monotonicity, global validity and negation.
\begin{lemma}\label{lemma:mon}
     If $\Vdash_{S}^{L} A$, $S$-monotonicity holds for $A$ and $A \Vdash_{S}^{G} B$, then both $\Vdash_{S}^{L} B$ and $\Vdash_{S}^{G} B$.
\end{lemma}

\begin{proof} Since $\Vdash_{S}^{L} A$ holds and monotonicity holds for $A$, for all $S'$ such that $S \subseteq S'$ we have that $\Vdash_{S'}^{L} A$. Since $A \Vdash_{S}^{G} B$ holds and $S$ is an extension of itself, we immediately conclude that  $\Vdash_{S'}^{L} B$ for all $S'$ extending $S$ and all $S''$ extending any $S'$, and thus $\Vdash_{S}^{G} B$. In particular, since $S$ is an extension of itself, we also have $\Vdash_{S}^{L} B$. 
\end{proof}

\begin{lemma}\label{lemma:neg}
 $(p \vdash_{S} \bot)$ iff $(p^i \Vdash_{S}^{L} \bot)$ iff $(p^i \Vdash_{S}^{G} \bot)$ iff $(\Vdash_{S}^{L} \neg p^i)$.

\end{lemma}
\begin{proof} It follows from Theorem~\ref{thm:monotonicity} that  $(p^i \Vdash_{S}^{L} \bot)$ iff $(p^i \Vdash_{S}^{G} \bot)$, since $p^i$ is an intuitionistic formula and $\bot$ is either the intuitionistic $\bot^{i}$ or the equivalent $\bot^{c}$. On the other hand, $(p^i \Vdash_{S}^{G} \bot)$ iff $(\Vdash_{S}^{L} \neg p^i)$ follows from the clause for implication. Thus we only need to prove that $(p \vdash_{S} \bot)$ iff $(p^i \Vdash_{S}^{L} \bot)$.

\medskip

\begin{enumerate}
    \item[$(\Rightarrow)$] Let $p \vdash_{S} \bot$. 
%For the sake of contradiction, 
Assume that there is an extension $S'$ of $S$ in which $\Vdash_{S'}^{L} p^i$. By Clause 1 we have $\vdash_{S'} p$ and, since $S'$ is an extension of $S$ and thus contains all its rules,  $p \vdash_{S'} \bot$ . By composing both derivations we get $\vdash_{S'} \bot $, which clashes with the consistency requirement. Thus, for all $S'$ extending $S$ we have $\nVdash_{S}^{L} p^i$, which together with Clause 4 yields $p^i \Vdash_{S}^{L} \bot$.

\medskip

\item[$(\Leftarrow)$] Assume $p^i \Vdash^{L}_{S} \bot$ and $p \nvdash_{S} \bot$. Let $S'$ be the system obtained by adding to $S$ only a rule $\alpha$ which concludes $p$ from empty premises. We start by proving that $S'$ is consistent, as thus a valid extension of $S$ in our semantics.
\end{enumerate}

Assume $S'$ is inconsistent, and consider the proof $\Pi$ of $\bot$ in $S'$. There are two possibilities:

\begin{enumerate}
    \item If $\Pi$ does not use the rule $\alpha$, then $\Pi$ is a proof in $S$. This yields a contradiction, as $S$ must be consistent.
    \item If $\Pi$ uses the rule $\alpha$, replace each application of $\alpha$ in $\Pi$ by an assumption $p$. This immediately yields a derivation showing $p \vdash_S \bot$, contradicting the second initial hypothesis.
\end{enumerate}

We then conclude that $S'$ is consistent. But, given that $\vdash_{S'} p$, we have $\Vdash_{S'}^{L} p^i$, which can be used together with the assumption $p^i \Vdash^{L}_{S} \bot$ to show $ \Vdash^{L}_{S} \bot$ and thus $\vdash_S \bot$, contradicting the consistency requirement. Thus $p \vdash_{S} \bot$.  
\end{proof}

\begin{corollary}\label{cor:class-iff-not-int}
    $\Vdash_{S}^{L} p^c$ iff $p^i \nVdash_{S}^{L} \bot$.
\end{corollary}

%\begin{proof} We divide the proof in two small parts:
%
%($\Rightarrow$) Assume $\Vdash_{S}^{L} p^c$. By clause 2, $p \nvdash_ S \bot$. By Lemma \ref{lemma:neg}, $p^i \nVdash_{S}^{L} \bot$.
%
%($\Leftarrow$) Assume $p^i \nVdash_{S}^{L} \bot$. By Lemma \ref{lemma:neg}, $p \nvdash_S \bot$. By clause 2, $\Vdash_{S}^{L} p^c$.  
%
%\end{proof}

\begin{lemma}\label{lemma:p^c implies existence of extension}
    $A \nVdash^L_{S} \bot$ iff there is some $S \subseteq S'$ such that $\Vdash_{S'}^{L} A$.
\end{lemma}

\begin{proof} %We divide the proof in two small steps:
% ($\Rightarrow$) 
  Assume $A \nVdash^L_{S} \bot$.  Suppose there is no $S \subseteq S'$ with $\Vdash_{S'}^{L} A$. Then $A \Vdash_{S}^{L} \bot$ holds vacuously, which is a contradiction. Hence, for some $S \subseteq S^n$ we have $\Vdash_{S'}^{L} A$.
 On the other hand, 
%  ($\Leftarrow$) 
assume that  there is some $S \subseteq S'$ such that $\Vdash_{S'}^{L} A$. Suppose $A \Vdash^L_{S} \bot$. Then we have $\Vdash_{S'}^{L} \bot$, yielding a contradiction. Thus, $A \nVdash^L_{S} \bot$.   
\end{proof}

%%%NOVO LEMA

\begin{lemma}\label{lemma:maximalconsistencyexistence}
Every consistent system $S$ has a $\bot$\textit{-complete} extension $S'$ such that, for every $p \in \Atb$, either $\vdash_{S'} p$ or $p \vdash_{S'} \bot$
\end{lemma}

\begin{proof}
Since $\Atb$ is countable, assign to each of its elements a unique natural number greater than $0$ as a superscript.  Define $S = S^{0}$, and consider the following construction procedure (for $m \geq 1)$:

\begin{enumerate}
    \item If $p^{m} \nvdash_{S^{m-1}} \bot$, then $S^{m}$ is obtained by adding a atomic axiom with conclusion $p^{m}$ to $S^{m-1}$;

    \medskip
        \item If $p^{m} \vdash_{S^{m-1}} \bot$, then $S^{m} = S^{m-1}$.
\end{enumerate}

We briefly show by a simple induction that every $S^{m}$ produced this way is consistent. Since $S^{0} = S$, $S^{0}$ is consistent. Now suppose that $S^{m-1}$ is consistent. If $p^{m} \vdash_{S^{m -1}} \bot$ then $S^{m-1} = S^{m}$, so $S^{m}$ is consistent. If $p^{m} \nvdash_{S^{m -1}} \bot$, assume for the sake of contradiction that $\vdash_{S^{m}} \bot$. Then either the proof of $\bot$ in $S^{m}$ does not use the atomic axiom with conclusion $p^{m}$ included in $S^{m-1}$, in which case it is also a deduction showing $\vdash_{S^{m-1}} \bot$, or it does use the atomic axiom, in which case by removing all instances of it from the deduction we obtain a deduction showing $p^{m} \vdash_{S^{m-1}} \bot$. In the first case we contradict the assumption that $S^{m-1}$ was consistent, and in the second we contradict the assumption that $p^{m} \nvdash_{S^{m-1}} \bot$, so in any case we obtain a contradiction. Hence, if $S^{m-1}$ is consistent then $S^{m}$ is consistent, and since $S^{0}$ is consistent we have that each $S^{m}$ is consistent.

Now let $S^{\bot  C} = \{R \in S^{m} | m \geq 1 \}$. Clearly, $S \subseteq S^{\bot C}$. To show that $S^{\bot C}$ is also consistent, assume for the sake of contradiction that there is a deduction showing $\vdash_{S^{\bot C}} \bot$. By the definition of deducibility, this deduction can only use finitely many rules. If the deduction does not use any atomic axioms, it is already a deduction in $S$, thus contradicting the fact that $S$ is consistent. If it does use atomic axioms, let $m$ be the greatest superscript occurring in atomic axioms of the deduction. This deduction only uses axioms with superscript equal to or less than $m$, thus all rules used in it must already occur in $S^{m}$ (as they could not have been added later in the construction). But then this means that $\vdash_{S^{m}} \bot$, contradicting our result that each $S^{m}$ is consistent. In both cases we reach a contradiction, so we conclude that $S^{\bot C}$ is indeed consistent.

Finally, pick any $p^{m} \in \Atb$. If $p^{m} \vdash_{S^{m-1}} \bot$ then $p^{m} \vdash_{S^{\bot C}} \bot$, as by the definition of $S^{\bot C}$ we have $S^{m - 1} \subseteq S^{\bot C}$. If $p^{m} \nvdash_{S^{m-1}} \bot$ then $S^{m}$ contains an atomic axiom concluding $p^{m}$, and since $S^{m} \subseteq S^{\bot C}$ we conclude $\vdash_{S^{\bot C}} p^{m}$. Since this holds for every $m \geq 1$ and every atom was assigned such a superscript, we conclude that for every $p \in \Atb$ either $p \vdash_{S^{\bot C}} \bot$ or $\vdash_{S^{\bot C}} p$, as desired.

\end{proof}

\begin{lemma}\label{lemma:sameextensionsproperty}
Let $A$ and $B$ be two formulas and $S$ a system such that, for all $S \subseteq S'$, $\Vdash^{L}_{S} A$ iff $\Vdash^{L}_{S'} A$ and $\Vdash^{L}_{S} B$ iff $\Vdash^{L}_{S'} B$. Then for all $S \subseteq S'$ we have $A \Vdash^{L}_{S'} B$ iff $A \Vdash^{G}_{S'} B$ iff either $\nVdash^{L}_{S'} A$ or $\Vdash^{L}_{S'} B$.
\end{lemma}

\begin{proof}
Let $S$ be any system $A$ and $B$ be two formulas as specified. Let $S'$ be an arbitrary extension of $S$. If $\nVdash^{L}_{S'} A$ then $\nVdash^{L}_{S} A$ and for all $S \subseteq S''$ we have $\nVdash^{L}_{S''} A$. Since $S' \subseteq S'''$ implies $S \subseteq S'''$ we have that $\nVdash^{L}_{S'''} A$  also holds for all extensions $S'''$ of $S'$, so $A \Vdash^{L}_{S'} B$ and $A \Vdash^{G}_{S'} B$ hold vacuously. Likewise, if $\Vdash^{L}_{S'} B$ then $\Vdash^{L}_{S} B$ and  so also $\Vdash^{L}_{S''} B$ for all $S \subseteq S''$, so again for all $S' \subseteq S'''$ we have $\Vdash^{L}_{S'''} B$, hence by similar reasoning we conclude that both $A \Vdash^{L}_{S'} B$ and $A \Vdash^{G}_{S'} B$ hold. Finally, if $\Vdash^{L}_{S'} A$ and $\nVdash^{L}_{S'} B$ hold then since $S' \subseteq S'$ clearly $A \nVdash^{L}_{S'} B$ and $A \nVdash^{G}_{S'} B$, and since this covers all cases we conclude the desired result.
     
\end{proof}

\begin{lemma}\label{lemma:maximalconsistentpersistency}
Let $S^{\bot C}$ be a $\bot$-complete extension of some system. Then, for every $S^{\bot C} \subseteq S'$ and every $A$, $\Vdash^{L}_{S^{\bot C}} A$ iff $\Vdash^{L}_{S'} A$.
\end{lemma}

\begin{proof}
    We show the result by induction on the complexity of formulas.

    \begin{enumerate}
   
      \item $A = p^{i}$.  

      \medskip

      ($\Rightarrow$) If $\Vdash^{L}_{S^{\bot C}} p^{i}$ then $\vdash_{S^{\bot C}} p$, and since the deduction showing this is also a deduction in the arbitrary $S'$ we have $\vdash_{S'} p$ and $\Vdash^{L}_{S'} p^{i}$. 

      \medskip
      
      ($\Leftarrow$) Assume $\Vdash^{L}_{S'} p^{i}$. Then $\vdash_{S'} p$. If we had a deduction showing $p \vdash_{S^{\bot C}} \bot$ then it would also be a deduction showing $p \vdash_{S'} \bot$, which would allow us to compose the deductions to show $\vdash_{S'} \bot$ and obtain a contradiction, hence $p \nvdash_{S^{\bot C}} \bot$. But by the definition of $\bot$-complete extensions we have that $p \nvdash_{S^{\bot C}} \bot$ implies $\vdash_{S^{\bot C}} p$, so we conclude $\Vdash^{L}_{S^{\bot C}} p^{i}$.

        \medskip
        
       \item $A = p^{c}$.  

       \medskip
       
    ($\Rightarrow$) If $\Vdash^{L}_{S^{\bot C}} p^{c}$ then $p \nvdash^{L}_{S^{\bot C}} \bot$, hence definition of $\bot$-complete extensions $\vdash_{S^{\bot C}} p$, so $\Vdash^{L}_{S^{'}} p^{i}$ for all $S^{\bot C} \subseteq S'$ and by Theorem \ref{lemma:intimpliesclasatom} and Clause 7 of Definition \ref{def:wvalidity} also $\Vdash^{L}_{S'} p^{c}$. 

    \medskip
       
      ($\Leftarrow$) Assume $\Vdash^{L}_{S'} p^{c}$ for arbitrary $S^{\bot C} \subseteq S'$. Then $p \nvdash_{S'} \bot$. If $p \vdash_{S^{\bot C}} \bot$ then since $S^{\bot C} \subseteq S'$ we would have $p \vdash_{S'} \bot$ and this would lead to a contradiction, so $p \nvdash_{S^{\bot C}} \bot$ and by definition of $\bot$-complete extension $\vdash_{S^{\bot C}} p$ and so $\Vdash^{L}_{S^{\bot C}} p^{i}$, which by Theorem \ref{lemma:intimpliesclasatom} and Clause 7 of Definition \ref{def:wvalidity} yields $\Vdash^{L}_{S^{\bot C}} p^{c}$.

        \medskip

       % \item $A = B^{C}$.  If $\Vdash^{L}_{S^{\bot C}} B^{c}$ then $B^{i} \nVdash^{L}_{S^{\bot C}} \bot$. By Lemma \ref{lemma:p^c implies existence of extension} there is a $S^{\bot C} \subseteq S''$ with $\Vdash^{L}_{S''} B^{i}$. Induction hypothesis: $\Vdash^{L}_{S''} B^{i}$ implies  $\Vdash^{L}_{S^{\bot C}} B^{i}$ implies $\Vdash^{L}_{S'} B^{i}$  for every $S^{\bot C} \subseteq S'$. Then by Theorem \ref{lemma:intimpliesclas} also $\Vdash^{L}_{S'} B^{c}$ for all $S^{\bot C} \subseteq S'$. For the converse, let  $\Vdash^{L}_{S'} B^{c}$ for arbitrary $S^{\bot C} \subseteq S'$. Then $B^{i} \nVdash^{L}_{S^{'}}  \bot$, hence by Lemma \ref{lemma:p^c implies existence of extension} there is a $S' \subseteq S''$ with $\Vdash^{L}_{S''} B^{i}$. Induction hypothesis: $\Vdash^{L}_{S''} B^{i}$ implies  $\Vdash^{L}_{S^{\bot C}} B^{i}$. Then by Theorem \ref{lemma:intimpliesclas} we have $\Vdash^{L}_{S^{\bot C}} B^{c}$. 

       \item $A = B^{C}$.  

       \medskip
       
     ($\Rightarrow$)  If $\Vdash^{L}_{S^{\bot C}} B^{c}$ then $B \nVdash^{L}_{S^{\bot C}} \bot$. By Lemma \ref{lemma:p^c implies existence of extension} there is a $S^{\bot C} \subseteq S''$ with $\Vdash^{L}_{S''} B^{i}$. Induction hypothesis:  $\Vdash^{L}_{S^{\bot C}} B^{i}$ and also $\Vdash^{L}_{S'} B^{i}$  for every $S^{\bot C} \subseteq S'$. Then by Theorem \ref{lemma:intimpliesclas} also $\Vdash^{L}_{S'} B^{c}$ for all $S^{\bot C} \subseteq S'$.

       \medskip
       
      ($\Leftarrow$) Let $\Vdash^{L}_{S'} B^{c}$ for arbitrary $S^{\bot C} \subseteq S'$. Then $B^{i} \nVdash^{L}_{S^{'}}  \bot$, hence by Lemma \ref{lemma:p^c implies existence of extension} there is a $S' \subseteq S''$ with $\Vdash^{L}_{S''} B^{i}$. Induction hypothesis: $\Vdash^{L}_{S^{\bot C}} B^{i}$. Then by Theorem \ref{lemma:intimpliesclas} we have $\Vdash^{L}_{S^{\bot C}} B^{c}$.

        \medskip

       \item $A = B \land C$. 

       \medskip

       ($\Rightarrow$) If $\Vdash^{L}_{S^{\bot C}} B \land C$ then $\Vdash^{L}_{S^{\bot C}} B$ and $\Vdash^{L}_{S^{\bot C}}C$. Induction hypothesis: $\Vdash^{L}_{S'} B$ and $\Vdash^{L}_{S'} C$ for all $S^{\bot C} \subseteq S'$. Then $\Vdash^{L}_{S'} B \land C$ for arbitrary $S^{\bot C} \subseteq S'$. 

       \medskip
       
      ($\Leftarrow$) Let $\Vdash^{L}_{S'} B \land C$ for arbitrary $S^{\bot C} \subseteq S'$. Then $\Vdash^{L}_{S'} B$ and $\Vdash^{L}_{S'} C$. Induction hypothesis: $\Vdash^{L}_{S^{\bot C}} B$ and $\Vdash^{L}_{S^{\bot C}} C$. Then $\Vdash^{L}_{S^{\bot C}} B \land C$.

\medskip
       
       \item $A = B \to C$. 
       
       \medskip
       
       ($\Rightarrow$) If $\Vdash^{L}_{S^{\bot C}} B \to C$ then $B \Vdash^{G}_{S^{\bot C}} C$. Induction hypothesis: for all $S^{\bot C} \subseteq S'$, $\Vdash^{L}_{S^{\bot C}} B$ iff $\Vdash^{L}_{S'} B$ and $\Vdash^{L}_{S^{\bot C}} C$ iff $\Vdash^{L}_{S'} C$. Then Lemma \ref{lemma:sameextensionsproperty} applies to $S^{\bot C}$ and all its extensions with respect to $B$ and $C$, so since $B \Vdash^{G}_{S^{\bot C}} C$ we conclude that either $\nVdash^{L}_{S^{\bot C}} B$ or $\Vdash^{L}_{S^{\bot C}} C$. The induction hypothesis then yields $\nVdash^{L}_{S'} B$ or $\Vdash^{L}_{S'} C$ for our chosen $S'$, so by Lemma \ref{lemma:sameextensionsproperty} we have $B \Vdash^{G}_{S'} C$ and thus $\Vdash^{L}_{S'} B \to C$. 

       \medskip
       
      ($\Leftarrow$) Assume $\Vdash^{L}_{S'} B \to C$. Then $B \Vdash^{G}_{S'} C$. Induction hypothesis: for all $S^{\bot C} \subseteq S'$, $\Vdash^{L}_{S^{\bot C}} B$ iff $\Vdash^{L}_{S'} B$ and $\Vdash^{L}_{S^{\bot C}} C$ iff $\Vdash^{L}_{S'} C$. Then again by applying Lemma \ref{lemma:sameextensionsproperty} from $B \Vdash^{G}_{S'} C$ we conclude $\nVdash^{L}_{S'} B$ or $\Vdash^{L}_{S'} C$, from the induction hypothesis we conclude $\nVdash^{L}_{S^{\bot C}} B$ or $\Vdash^{L}_{S^{\bot C}} C$ and by Lemma \ref{lemma:sameextensionsproperty} we conclude $B \Vdash^{G}_{S^{\bot C}} C$, sowe conclude $\Vdash^{L}_{S^{\bot C}} B \to C$.

     %  But then $B$ and $C$ hold in some extension $S'$ of $S^{\bot C}$ if and only if they also holds in all extensions of $S'$ (since $S' \subseteq S''$ implies $S^{\bot C} \subseteq S''$), so it's easy to check that for any $S^{\bot C} \subseteq S'$ we have that $B \Vdash^{G}_{S'} C$ if and only if  $\nVdash^{L}_{S'} B$ or $\Vdash^{L}_{S'} C$. Since we had already concluded that $B \Vdash^{G}_{S^{\bot C}} C$ holds and $S^{\bot C}$ is an extension of itself we conclude $\nVdash^{L}_{S^{\bot C}} B$ or $\Vdash^{L}_{S^{\bot C}} C$. The first case yields $\nVdash^{L}_{S'} B$ for the arbitrary $S'$ and the second yields $\Vdash^{L}_{S'} C$, so in both cases we have $B \Vdash^{G}_{S'} C$ and thus $\Vdash^{L}_{S'} B \to C$. For the converse, assume $\Vdash^{L}_{S'} B \to C$. Then $B \Vdash^{G}_{S'} C$. Induction hypothesis: $\Vdash^{L}_{S^{\bot C}} B$ iff $\Vdash^{L}_{S''} B$ and $\Vdash^{L}_{S^{\bot C}} C$ iff $\Vdash^{L}_{S''} C$ for arbitrary $S^{\bot C} \subseteq S''$. Once again we have that $B \Vdash^{G}_{S'} C$ if and only if  $\nVdash^{L}_{S'} B$ or $\Vdash^{L}_{S'} C$, so since $B \Vdash^{G}_{S'} C$ hold we conclude that either  $\nVdash^{L}_{S'} B$ or $\Vdash^{L}_{S'} C$. Since $\nVdash^{L}_{S'} B$ implies $\nVdash^{L}_{S^{\bot C}} B$ and $\Vdash^{L}_{S'} C$ implies $\Vdash^{L}_{S^{\bot C}} C$ and in both cases we have $B \Vdash^{G}_{S^{\bot C}} C$, we finally conclude  $\Vdash^{L}_{S^{\bot C}} B \to C$.

       \medskip

       \item $A = B \lor C$. 

       \medskip
       
      ($\Rightarrow$) If $\Vdash^{L}_{S^{\bot C}} B \lor C$ then, for all $S^{\bot C} \subseteq S'$ and all $p \in \Atb$, if $A \Vdash^{L}_{S'} p^{i}$ and $B \Vdash^{L}_{S'} p^{i}$ then  $\Vdash^{L}_{S'} p^{i}$. Now pick any $S^{\bot C} \subseteq S'$, pick any $p$ and let $S''$ be any extension of $S'$ with $A \Vdash^{L}_{S''} p^{i}$ and $B \Vdash^{L}_{S''} p^{i}$. Then since $S' \subseteq S''$ implies $S^{\bot C} \subseteq S''$ we conclude $\Vdash^{L}_{S'} p^{i}$, so by arbitrariness of $S''$ already $\Vdash^{L}_{S'} B \lor C$. 

      \medskip
       
      ($\Leftarrow$) Assume $\Vdash^{L}_{S'} B \lor C$. Then, for all $p \in \Atb$, if $A \Vdash^{L}_{S'} p^{i}$ and $B \Vdash^{L}_{S'} p^{i}$ then $\Vdash^{L}_{S'} p^{i}$ (which is a special case of the semantic condition for disjunction). Induction hypothesis: for all $S^{\bot C} \subseteq S''$, $\Vdash^{L}_{S^{\bot C}} B$ iff $\Vdash^{L}_{S''} B$, $\Vdash^{L}_{S^{\bot C}} C$ iff $\Vdash^{L}_{S''} C$ and $\Vdash^{L}_{S^{\bot C}} p^{i}$ iff $\Vdash^{L}_{S''} p^{i}$ for all $p \in \Atb$. Now pick any $S^{\bot C} \subseteq S^{''}$ such that $A \Vdash^{L}_{S''} p^{i}$ and $B \Vdash^{L}_{S''} p^{i}$ for a particular $p \in \Atb$, if any. Then, by the induction hypothesis, since both $S'$ and $S''$ are extensions of $S^{\bot C}$ we have $\Vdash^{L}_{S''} B$ iff $\Vdash^{L}_{S'} B$, $\Vdash^{L}_{S''} C$ iff $\Vdash^{L}_{S'} C$ and $\Vdash^{L}_{S''} p^{i}$ iff $\Vdash^{L}_{S'} p^{i}$, so since $A \Vdash^{L}_{S''} p^{i}$ and $B \Vdash^{L}_{S''} p^{i}$ we can apply Lemma \ref{lemma:sameextensionsproperty} to conclude based only on the equivalnce of $B$, $C$ and $p^{i}$ in $S''$ and $S'$ that $A \Vdash^{L}_{S'} p^{i}$ and $B \Vdash^{L}_{S'} p^{i}$ also hold, hence since  $\Vdash^{L}_{S'} B \lor C$ we conclude $\Vdash^{L}_{S'} p^{i}$. Since $\Vdash^{L}_{S'} p^{i}$ by the induction hypothesis we conclude $\Vdash^{L}_{S''} p^{i}$, hence for any $S^{\bot C} \subseteq S''$ we have that if  $A \Vdash^{L}_{S''} p^{i}$ and $B \Vdash^{L}_{S''} p^{i}$ for an arbitrary $p \in \Atb$ then $\Vdash^{L}_{S''} p^{i}$, which yields $\Vdash^{L}_{S^{\bot C}} B \lor C$.

    \end{enumerate}

\end{proof}

\subsection{Weak ecumenical behaviour}\label{sec:eb}
%Now we will proceed to prove some results concerning the behaviour of classical formulas in this system.

This section will be devoted to show some interesting behaviours when monotonicity does not hold for ecumenical formulas. Notice that, due to Corollary~\ref{cor:class-iff-not-int}, classical atoms $p^c$ and classical non-atomic formulas $A^c$ may be treated uniformly in some cases.

\begin{theorem}\label{thm:twoglobalequivalences}
     $A^c \Vdash \neg \neg A^i$ and $\neg \neg A^i \Vdash A^c$.
\end{theorem}

\begin{proof} Let's first prove  that $A^c \Vdash_{S}^{G} \neg \neg A^i$ holds for arbitrary $S$.

Let $S$ be an arbitrary atomic system. Let $S'$ be any extension of $S$ such that, for all $S' \subseteq S'' $, we have $ \Vdash_{S''}^{L} A^c$. Then for every $S' \subseteq S''$ we have $A^i \nVdash_{S''}^{L} \bot$.  Suppose, for the sake of contradiction, that there is a $S' \subseteq S''$ such that  $\Vdash_{S''}^{L} \neg A^i$. Then $A^{i} \Vdash^{G}_{S''} \bot$. Now let $S^{\bot C}$ be a $\bot$-complete extension of $S''$. Since $S' \subseteq S^{\bot C}$ we have $A^i \nVdash_{S^{\bot C}}^{L} \bot$, hence by Lemma \ref{lemma:p^c implies existence of extension} there must be a $S^{\bot C} \subseteq S'''$ with $\Vdash^{L}_{S'''} A$. So by Lemma \ref{lemma:maximalconsistentpersistency} we conclude $\Vdash^{L}_{S^{\bot C}} A$ and also $\Vdash^{L}_{S'''} A$ for arbitrary extensions $S'''$ of $S^{\bot C}$. But since $S'' \subseteq S^{\bot C}$, $A \Vdash^{G}_{S''} \bot$ and $\Vdash^{L}_{S'''} A$ for every $S^{\bot C} \subseteq S'''$ we conclude  $\Vdash^{L}_{S'''} \bot$ and $\vdash_{S'''} \bot$ for all $S^{\bot C} \subseteq S'''$, which violates the consistency requirement. Therefore, for all $S' \subseteq S''$ we have  $\nVdash_{S''}^{L} \neg A^i$ and so by Lemma \ref{lemma:newsimplification} also $\Vdash^{L}_{S''} \neg \neg  A^{i}$, hence by arbitrariness of $S'$ we have $A^c \Vdash_{S}^{G} \neg \neg A^i$.

Now, let's prove $\neg \neg A^i \Vdash A^c$, which amounts to proving $\neg \neg A^i \Vdash_{S}^{G} A^c$ for arbitrary $S$.
Let $S$ be an arbitrary atomic system. Let $S'$ be any extension of $S$ such that, for all $S''$ for which $S' \subseteq S'' $, we have that $\Vdash_{S''}^{L} \neg \neg A^i$ holds. In particular, $\neg A^{i} \Vdash^{G}_{S'} \bot$ holds. Now assume for the sake of contradiction that, for some $S' \subseteq S''$, we have $ A^i \Vdash_{S''}^{L}  \bot$. Assume there is a $S'' \subseteq S'''$ such that $\Vdash^{L}_{S'''} A^{i}$. Then since $ A^i \Vdash_{S''}^{L}  \bot$ and $S'' \subseteq S'''$ we conclude  $\Vdash_{S'''}^{L}  \bot$ and obtain a contradiction, so for all $S'' \subseteq S'''$ we have $\nVdash^{L}_{S'''} A^{i}$. Since $\nVdash^{L}_{S'''} A^{i}$ for all $S'' \subseteq S'''$ we have $\Vdash^{L}_{S'''} \neg A^{i}$ for all $S'' \subseteq S'''$ by Lemma \ref{lemma:newsimplification}. But since $\neg A^{i} \Vdash^{G}_{S'} \bot$, $S' \subseteq S''$ and $\Vdash^{L}_{S'''} \neg A^{i}$ holds for all $S'' \subseteq S'''$ we conclude $\Vdash^{L}_{S'''} \bot$ for all such $S'''$, leading to a contradiction. Hence we conclude that there can be no extension $S''$ of $S'$ with $ A^i \Vdash_{S''}^{L}  \bot$, so for all $S' \subseteq S''$ we have $ A^i \nVdash_{S''}^{L}  \bot$ and thus $\Vdash_{S''}^{L} A^{c}$, hence since $S'$ was arbitrary we conclude $\neg \neg A^i \Vdash_{S}^{G} A^c$.

\end{proof}

The next two results are interesting, showing that global validity can be preserved locally but that this is not always the case.

\begin{theorem}\label{thm:Acai}
     $A^c \Vdash_{S}^{L} \neg \neg A^i$ does not hold for arbitrary $S$.
\end{theorem}

\begin{proof}
We prove the result for atoms. Consider the atomic system $\emptyset$, which contains no rules. Clearly, since $p \nvdash_{\emptyset} \bot$, we have $\Vdash_{\emptyset}^{L} p^c$. Suppose that $\Vdash_{\emptyset}^{L} \neg \neg p^i$. Consider now an extension $S$ of $\emptyset$ containing a rule which concludes $\bot$ from the premise $p$. Hence $p \vdash_S \bot$ and, due to Lemma~\ref{lemma:neg}, $\Vdash_{S}^{L} \neg p^i$ holds. Since $\neg \neg p^i$ is intuitionistic, it is monotonic, and thus $\Vdash_{\emptyset}^{L} \neg \neg p^i$ implies $\Vdash_{S}^{L} \neg \neg p^i$. By the semantic clause for implication we then have $\neg p^i \Vdash_{S}^{G} \bot$ and, since $\Vdash_{S}^{L} \neg p^i$,  by Lemma~\ref{lemma:mon}  we have $\Vdash_{S}^{L} \bot$, and thus $\vdash_S \bot$. Contradiction. Thus $\nVdash_{\emptyset}^{L} \neg \neg p^i$ and, since the empty set is an extension of itself, $p^c \nVdash_{\emptyset}^{L} \neg \neg p^i$. 
  
 \end{proof}

\begin{theorem} \label{thm:ic}
     $\neg \neg A^i \Vdash_{S}^{L} A^c$ holds for arbitrary $S$.
\end{theorem}

\begin{proof}

Let $S$ be any system. Consider any $S \subseteq S'$ such that $\Vdash_{S'}^{L} \neg \neg A^i$. By the clause for implication, $\neg A^i \Vdash_{S'}^{G} \bot$. Assume $A^i \Vdash_{S'}^{L} \bot$ for the sake of contradiction. Then clearly $\nVdash^{L}_{S''} A^{i}$ for all $S' \subseteq S''$, so $\Vdash^{L}_{S''} \neg A^{i}$ for all $S' \subseteq S''$ by Lemma \ref{lemma:newsimplification}. Since  $\neg A^i \Vdash_{S'}^{G} \bot$ and $\Vdash^{L}_{S''} \neg A^{i}$ for all $S' \subseteq S''$ we conclude $\Vdash^{L}_{S''} \bot$ for all $S' \subseteq S''$, which is a contradiction. Hence $A^i \nVdash_{S'}^{L} \bot$, so $\Vdash_{S'}^{L} A^{c}$, therefore by arbitrariness of $S'$ we have $\neg \neg A^{i} \Vdash^{L}_{S} A^{c}$.

%Then we have $A^i \Vdash_{S}^{G} \bot$ by Lemma \ref{lemma:localimpliesglobal},  $\Vdash_{S'}^{L} \neg A^i$ by the clause for implication and $\Vdash_{S'}^{G} \neg A^i$ by another application of Lemma \ref{lemma:localimpliesglobal}. This yields $\Vdash_{S'}^{G} \bot$, which leads to a contradiction. Then we conclude $A^i \nVdash_{S'}^{L} \bot$, and thus $\Vdash_{S'}^{L} A^c$. 

%Assume, for the sake of contradiction, that there is an extension $S^j$ of $S^n$ such that $\nVdash_{S^j}^{L} A^c$. By Clause 2 we have $A^i \Vdash_{S^j}^{L} \bot$. By Lemma \ref{lemma:localimpliesglobal} we have $A^i \Vdash_{S^j}^{G} \bot$, and thus, we have $\Vdash_{S^j}^{L} \neg A^i$. Since $S^j$ is a extension of $S^n$ and monotonicity holds for $\neg \neg A^i$, we have $\Vdash_{S^j}^{L} \neg \neg A^i$, and thus $\neg p^i \Vdash_{S^j}^{G} \bot$. By Lemma~\ref{lemma:mon}, this yields $\Vdash_{S^j}^{L} \bot$, and thus $\vdash_{S^j} \bot$. Contradiction. As such, $\Vdash_{S^j}^{L} p^c$ for any $S^j$, and thus $\neg \neg p^i \Vdash_{S}^{L} p^c$ holds for any $S$.  
\end{proof}
\begin{remark}
Put together, these results show that classical proof of $A$ is strictly weaker than an intuitionistic proof of $\neg \neg A$, and justify the motto presented in the Introduction.
\end{remark}
%We could, to put it bluntly, state the relation between the two as:
%\begin{quote}
%{\em Classical proof plus monotonicity equals intuitionistic proof of double negation.}
%\end{quote}

The following results present ecumenical versions of the excluded middle and Peirce's law.

\begin{theorem}\label{thm:mp}
     $\Vdash_{S}^{L} A^c \lor \neg A^i$ holds for arbitrary $S$.
\end{theorem}

\begin{proof} 

Let $S$ be any system. Let $S'$ be any extension of $S$ in which $A^c \Vdash_{S'}^{L} p^i$ and $\neg A^{i} \Vdash_{S'}^{L} p^i$for some $p \in \Atb$. If $A^{i} \Vdash^{L}_{S'} \bot$ then by Lemma \ref{lemma:localimpliesglobal} we have $A^{i} \Vdash^{G}_{S'} \bot$ and so $\Vdash^{L}_{S'} \neg A^{i}$, hence since $\neg A^{i} \Vdash_{S'}^{L} p^i$ we conclude $\Vdash^{L}_{S'} p^{i}$. If $A \nVdash^{L}_{S'} \bot$ then $\Vdash^{L}_{S'} A^{c}$, so since $A^c \Vdash_{S'}^{L} p^i$ we conclude $\Vdash^{L}_{S'} p^{i}$. Since either $A \Vdash^{L}_{S'} \bot$ or $A \nVdash^{L}_{S'} \bot$ holds this cover all possible cases, hence for all $S \subseteq S'$ we have that $A^c \Vdash_{S'}^{L} p^i$ and $\neg A^{i} \Vdash_{S'}^{L} p^i$ implies $\Vdash^{L}_{S'} p^{i}$ for any $p \in \Atb$, so we conclude $\Vdash^{L}_{S} A^{c} \lor \neg A^{i}$.

\end{proof}

This proof is particularly interesting because it combines our notion of classical proof with the weak definition of disjunction (in terms of local validity) to provide a simple proof of the excluded middle in the language through an application of the excluded middle in the metalanguage. This suggests that the weak disjunction here proposed becomes predominantly classical if combined with our notion of classical proof, and this makes it so that other classical results could also be proved via metalinguistical applications of the excluded middle. This would not be the case if we were to define disjunction through global validity; as will be seen in the strong ecumenical semantics, this strenghtened disjunction has much more of a intuitionistic flavour.

\begin{theorem}\label{thm:pl}
     $\Vdash_{S}^{L} ((A^i \to B) \to A^i) \to A^c$ holds for arbitrary $S$.
\end{theorem}

\begin{proof}

Let $S$ be a system. Let $S'$ be an extension of $S$ with $\Vdash_{S''}^{L} (A^i \to B) \to A^i$ for all $S' \subseteq S''$. Then, by definition, for all those $S''$ we have $(A^i \to B) \Vdash_{S''}^{G} A^i$. Assume, for the sake of contradiction, that for some of those $S''$ we have $\nVdash_{S''}^{L} A^c$. Then we have $A^i \Vdash_{S''}^{L} \bot$ by Clause 3 of Definition \ref{def:wvalidity}, and since in any $S'' \subseteq S'''$ with $\Vdash^{L}_{S'''} A^{i}$ we could obtain  $\Vdash^{L}_{S'''} \bot$ and thus a contradiction we clearly have $\nVdash_{S'''}^{L} A^{i}$ for all $S'' \subseteq S'''$. But notice that, for any such $S'''$, $S''' \subseteq S''''$ implies $S'' \subseteq S''''$ and so $\nVdash_{S''''}^{L} A^{i}$, hence we have that $A^{i} \Vdash^{G}_{S'''} B$ is vacuously satisfied in all such $S'''$, so we conclude that for all $S'' \subseteq S'''$ we have $\Vdash^{L}_{S'''} A^{i} \to B$. Since $(A^i \to B) \Vdash_{S''}^{G} A^i$ and for all $S'' \subseteq S'''$ it holds that $\Vdash^{L}_{S'''} A^{i} \to B$ we conclude that for all $S'' \subseteq S'''$ we have $\Vdash^{L}_{S'''} A^{i}$ and, in particular, $\Vdash^{L}_{S''} A^{i}$. But we had previously concluded from our assumption for contradiction that $A^i \Vdash_{S''}^{L} \bot$, so since $\Vdash^{L}_{S''} A^{i}$ we have $\Vdash^{L}_{S''} \bot$, which is indeed a contradiction. Hence we conclude that for no $S' \subseteq S''$ we have $\nVdash_{S''}^{L} A^c$, so $\Vdash_{S''}^{L} A^c$ holds for all $S' \subseteq S''$. Since $S'$ is an arbitrary extension of $S$ with $\Vdash_{S''}^{L} (A^i \to B) \to A^i$ holding for all $S' \subseteq S''$ and we have shown that $\Vdash^{L}_{S''} A^{c}$ also holds for all such $S''$ we conclude $(A^i \to B) \to A^i \Vdash_{S}^{G} A^{c}$, and thus $\Vdash_{S}^{L} ((A^i \to B) \to A^i) \to A^{c}$.

%Let $S$ be any system. Let $S'$ be any extension of $S$ such that, for all $S''$ for which $S' \subseteq S''$ holds, $\Vdash_{S'}^{L} (A^i \to B) \to A^i$ holds. Then, by definition, for all those $S''$ we have $(A^i \to B) \Vdash_{S''}^{G} A^i$. Assume, for the sake of contradiction, that for some of those $S''$ we have $\nVdash_{S''}^{L} A^c$. Then we have $A^i \Vdash_{S''}^{L} \bot$ by Clause 3 and Corollary \ref{cor:class-iff-not-int}, and thus $A^i \Vdash_{S''}^{G} \bot$ by Lemma \ref{lemma:localimpliesglobal}.

%Then, since $A^i \Vdash_{S''}^{G} \bot$ and for all $S'' \subseteq S'''$ we have $\nVdash_{S'''}^{G} \bot$, we get $\nVdash_{S'''}^{G} A^i$ for every $S'''$ using the contrapositive of Lemma \ref{lemma:globalmodusponens}, a fact that can be combined with the clause for implication to show $\Vdash_{S''}^{L} A^i \to B$ for any $B$, and thus $\Vdash_{S''}^{G} A^i \to B$ by Lemma \ref{lemma:localimpliesglobal}. But we also have $\Vdash_{S''}^{L} (A^i \to B) \to A^i$ for all $S''$, and thus $(A^i \to B) \Vdash_{S''}^{G} A^i$ holds for this arbitrary $S''$, which yields $\Vdash_{S''}^{G} A^i$ via Lemma~\ref{lemma:globalmodusponens}. Since $\Vdash_{S''}^{G} A^i$  and $A^i \Vdash_{S''}^{G} \bot$, we also have $\Vdash_{S''}^{G} \bot$ by Lemma~\ref{lemma:globalmodusponens}, which yields the desired contradiction. Then, for every such $S''$ extending $S'$ we have $\Vdash_{S''}^{L} A^c$, and thus $(A^i \to B) \to A^i \Vdash_{S'}^{G} A^c$, which proves $\Vdash_{S}^{L} ((A^i \to B) \to A^i) \to A^c$ for arbitrary $S$.   
\end{proof}

There are, however, some drawbacks to our definitions, which are mainly due to the interaction between the clause for disjunction and the definition of local validity. 
For instance, we lose validities such as the following.
\begin{proposition}
    $ (A \lor B), (A \to C), (B \to C) \Vdash C$ does not hold in general.
\end{proposition}

\begin{proof}

Let $A = (\neg p^{i})$, $B = (p^c)$ and $C = (\neg p^i \lor \neg \neg p^i)$. Theorem \ref{thm:mp} has already shown that $(A \lor B) = \neg p^{i} \lor p^{c}$ is valid in all $S$, so now we show that this is also the case for $(A \to C) = \neg p \to (\neg p^i \lor \neg \neg p^i)$ and for $(B \to C) = p^{c} \to (\neg p^i \lor \neg \neg p^i)$. 

Let $S$ be any system and $S'$ any extension of it with $\Vdash^{L}_{S'}  \neg p^{i}$. Since $\neg p^{i}$ only containes intuitionistic subformulas Theorem \ref{thm:monotonicity} shows that it is monotonic, so by Theorem \ref{thm:monotoniccollapse} for all $S' \subseteq S''$ we have $\Vdash^{L}_{S''}  \neg p^{i}$. Then for any $S' \subseteq S''$ and any $q \in \Atb$ if both $ \neg p^{i} \Vdash^{L}_{S''} q^{i}$ and $\neg \neg p^{i} \Vdash^{L}_{S''} q^{i}$ we can combine $\neg p^{i} \Vdash^{L}_{S''} q^{i}$ with $\Vdash^{L}_{S''} \neg p^{i}$ to obtain $\Vdash^{L}_{S''} q^{i}$, so $\Vdash^{L}_{S'} \neg p^{i} \lor \neg \neg p^{i}$. Since $S'$ in arbitrary extension of $S$ with $\Vdash^{L}_{S'} \neg p^{i}$ we conclude $\neg p^{i} \Vdash^{L}_{S}  \neg p^{i} \lor \neg \neg p^{i}$ and thus also $\neg p^{i} \Vdash^{G}_{S}  \neg p^{i} \lor \neg \neg p^{i}$ by Lemma \ref{lemma:localimpliesglobal}, so we conclude $\Vdash^{L}_{S}  \neg p^{i} \to (\neg p^{i} \lor \neg \neg p^{i})$.

Now let $S$ be any system and $S'$ any extension of it with $\Vdash^{L}_{S''}  p^{c}$ for every $S' \subseteq S''$. Then for every such $S''$ we have $p \nvdash_{S''} \bot$. Assume that for some $S''$ we have $\Vdash^{L}_{S''} \neg p^{i}$. Then by Theorem \ref{lemma:neg} we have $p \vdash_{S''} \bot$, which yields a contradiction. Hence, for all such $S''$ we have $\nVdash^{L}_{S''} \neg p^{i}$, so by Lemma \ref{lemma:newsimplification} we get $\Vdash^{L}_{S''} \neg \neg p^{i}$. Now pick any such $S''$ and consider a $S'' \subseteq S'''$ with $\neg p^{i} \Vdash^{L}_{S'''} q^{i}$ and $\neg \neg p^{i} \Vdash^{L}_{S'''} q^{i}$ for some $q \in \Atb$. Since $\neg \neg p^{i}$ only contains intuitionistic subformulas we conclude by Theorem \ref{thm:monotonicity} that it is monotonic. Since $S'' \subseteq S'''$ and $\Vdash^{L}_{S''} \neg \neg p^{i}$, by Theorem \ref{thm:monotoniccollapse}  we have $\Vdash^{L}_{S'''} \neg \neg p^{i}$, hence $\neg \neg p^{i} \Vdash^{L}_{S'''} q^{i}$ and so $\Vdash^{L}_{S'''} q^{i}$. From this we conclude that $\Vdash^{L}_{S''} \neg p^{i} \lor \neg \neg p^{i}$. But notice that $S''$ is an arbitrary extension of $S'$, so we conclude that if $\Vdash^{L}_{S''}  p^{c}$ for every $S' \subseteq S''$ then $\Vdash^{L}_{S''} \neg p^{i} \lor \neg \neg p^{i}$ for all $S' \subseteq S''$, and by arbitrariness of $S'$ we conclude $p^{c} \Vdash^{G}_{S} \neg p^{i} \lor \neg \neg p^{i}$, and thus $\Vdash^{L}_{S} p^{c} \to ( \neg p^{i} \lor \neg \neg p^{i})$.

We have shown that $\Vdash^{L}_{S} A \lor B$, $\Vdash^{L}_{S} A \to C$ and $\Vdash^{L}_{S} B \to C$ hold in any $S$ for our choice of $A$, $B$ and $C$. Therefore, it suffices to show for some particular $S$ that $\nVdash^{L}_{S} \neg p^{i} \lor \neg \neg p^{i}$ (that is, $\nVdash^{L}_{S} C$) to prove the desired result.  But since both $\neg p^{i}$ and $\neg \neg p^{i}$ are purely intuitionistic formulas their semantics is identical to intuitionistic $\Bes$, so we can simply point out that $\neg p^{i} \lor \neg \neg p^{i}$ is not an intuitionistic theorem to conclude the desired result. In particular, $\nVdash^{L}_{\emptyset} \neg p^{i} \lor \neg \neg p^{i}$ (as only intuitionistic theorems hold in the empty system), so we have $ (A \lor B), (A \to C), (B \to C) \nVdash^{G}_{\emptyset} C$ and also $ (A \lor B), (A \to C), (B \to C) \nVdash C$.

\end{proof}

This issue seems to be caused by some interactions between the definition of implication (which uses global validity) and disjunction (which uses local validity). We could, of course, provide a weaker definition of implication that uses local validity, but since logical validity is defined by recourse to global validity by doing we would be giving away the deduction theorem (e.g. $p^{c} \Vdash^{L}_{\emptyset} \neg \neg p^{i}$ does not hold by Theorem \ref{thm:Acai} but $p^{c} \Vdash^{G}_{S} \neg \neg p^{i}$ holds for all $S$ by Theorem \ref{thm:twoglobalequivalences}, so we would have $\nVdash^{L}_{S} p^{c} \to \neg \neg p^{i}$ and so $\nVdash p^{c} \to \neg \neg p^{i}$ even though $p^{c} \Vdash \neg \neg p^{i}$). 

As seen in the proof of Theorem \ref{thm:mp}, for all $S$ and $A$ either we have $A^{i} \Vdash^{L}_{S} \bot$ and so $A^{i} \Vdash^{G}_{S} \bot$ by Lemma \ref{lemma:localimpliesglobal} thus also $\Vdash^{L}_{S} \neg A^{i}$ or $A^{i} \nVdash^{L}_{S} \bot$ and so $\Vdash^{L}_{S} \neg A^{c}$. In other words, the metalinguistic excluded middle is ``locally valid" in every $S$ due to our definition of classical proofs. On the other hand, even though from $A^{i} \Vdash^{G}_{S} \bot$ we can conclude $\Vdash^{L}_{S} \neg A^{i}$ it is not the case that from $A^{i} \nVdash^{G}_{S} \bot$ we can conclude $A^{i} \nVdash^{L}_{S} \bot$, which would be necessary for us to conclude $\Vdash^{L}_{S} \neg A^{c}$ (remember that Lemma \ref{lemma:globaltheoremimplieslocaltheorem} fails for non-empty contexts), so the metalinguistic excluded middle is not ``globally valid". This creates a certain tension between local and global definitions, as local definitions are able to draw on the local excluded middle to validate classical behavior but global definitions are not.

Although the semantic tension and the independent coexistence of classical and intuitionistic features are certainly desirable in the context of ecumenical semantics, the main issue with the definitions we have presented is that, since the usual rule for disjunction elimination is no longer sound, the weak ecumenical semantics is not easily captured in simple syntactic systems. This makes it so that the main motive for studying it lies in the clarification of the ways in which the global and local notions of validity relate to intuitionistic and classical concepts of proof. There might, of course, also be other combinations of local and global definitions which lead to interesting new ecumenical versions of $\Bes$, but we leave the study of any such combinations to future works.

The clarifications provided by the weak semantics on how the notion of classical proof behaves in $\Bes$ allow us to formulate a new kind of ecumenical semantics which fixes some of its issues. As such, we propose next an an ecumenical $\Bes$ with some stronger definitions and very different semantic properties.

%To fix the perceived weaknesses of the weak semantics, we propose next an ecumenical $\Bes$ with a stronger clause for classical formulas.

\section{Strong ecumenical semantics}\label{sec:ses}

In the weak semantics we define that a formula has a classical proof in $S$ if and only if it is consistent in $S$. As a result, classical proofs are not monotonic, so we need to differentiate between local and global validity notions. But there is another possibility: we can define that a formula has a classical proof in $S$ if and only if it is consistent in $S$ \textit{and all its extensions}. This is still faithful to Hilbert's ideas concerning classical proofs and truth, and since we only consider extensions of $S$ it is also faithful to the proposal of the original $\Bes$.

\begin{definition}[Strong Validity]\label{def:Validity} Strong $S$-validity and strong validity are defined as follows.
    
\begin{enumerate}
    \item $\vDash_{S} p^i$ iff $\vdash_S p$, for  $p \in\Atb$;

    \medskip    
    
    \item $\vDash_{S} p^c$ iff $\forall S' (S \subseteq S'): p \nvdash_{S'} \bot$, for  $p \in \Atb$;

\medskip    
    
    \item $\vDash_{S} A^c$ iff $\forall S' (S \subseteq S'): (A)^i \nvDash_{S'} \bot$, for $A \notin \Atb$;

 \medskip     
        
    \item $\vDash_{S} (A \land B)^i$ iff $\vDash_{S} A$ and $\vDash_{S} B$;

\medskip      

     \item $\vDash_{S} (A \to B)^i$ iff $A \vDash_{S}  B$; 

\medskip       

      \item $\vDash_{S} (A \lor B)^i$ iff  $ \forall S' (S \subseteq S') $ and all $p \in \Atb$: $A \vDash_{S'} p^i$ and $B \vDash_{S'} p^i$ implies $\vDash_{S'}  p^i$;
      
\medskip  

\item For non-empty $\Gamma$, we have that $\Gamma \vDash_{S} A$ iff for all $S'$ such that $S \subseteq S'$ it holds that, if $\vDash_{S'} B$ for all $B \in \Gamma$, then $\vDash_{S'} A$;

\medskip

      \item $\Gamma \vDash A $ iff $\Gamma \vDash_{S} A$ for all $S$.
\end{enumerate}
\end{definition}

Weak validity uses a non-monotonic notion, whereas in strong validity classical validities are monotonic by definition. Since by Theorem \ref{thm:monotoniccollapse} $S$-monotonicity induces a collapse between $\Vdash_{S}^{L}$ and $\Vdash_{S}^{G}$ and all formulas of the strong ecumenical semantics are monotonic,  local and global validities are non-distinguishable.

\section{An ecumenical proof system for strong ecumenical validity} \label{sec:ND}
In this section we will prove soundness and completeness of the natural deduction ecumenical system $\NE$ presented in Fig.~\ref{fig:NEB} (which is a version of the system $CIE$ presented in \cite{VictorEcumenical} with a restriction on iterations of the ``classicality'' operator)
w.r.t. the 
strong ecumenical $\Bes$. 

We say that $\Gamma \vdash_{\NE} A$ holds if and only if there is a deduction of $A$ from $\Gamma$ using the rules of $\NE$.

\begin{figure}[ht]
\begin{tabular}{lc@{\quad}l}
\end{tabular}
\begin{prooftree}
\AxiomC{$\Gamma\;[A]$}
\noLine
\UnaryInfC{$\Pi$}
\noLine
\UnaryInfC{$B$}
\RightLabel{\scriptsize{$\to$-int}}
\UnaryInfC{$A \to B$}
\DisplayProof
\qquad
\AxiomC{$\Gamma_1$}
\noLine
\UnaryInfC{$\Pi_1$}
\noLine
\UnaryInfC{$A \to B$}
\AxiomC{$\Gamma_2$}
\noLine
\UnaryInfC{$\Pi_2$}
\noLine
\UnaryInfC{$A$}
\RightLabel{\scriptsize{$\to$-elim}}
\BinaryInfC{$B$}
\end{prooftree}

\begin{prooftree}
\AxiomC{$\Gamma$}
\noLine
\UnaryInfC{$\Pi$}
\noLine
\UnaryInfC{$A_j$}
\RightLabel{\scriptsize{$\vee_j$-int}}
\UnaryInfC{$A_1\vee A_2$}
\DisplayProof
\qquad
\AxiomC{$\Gamma_1$}
\noLine
\UnaryInfC{$\Pi_1$}
\noLine
\UnaryInfC{$A \vee B$}
\AxiomC{$\Gamma_2\;[A]$}
\noLine
\UnaryInfC{$\Pi_2$}
\noLine
\UnaryInfC{$C$}
\AxiomC{$\Gamma_3\;[B]$}
\noLine
\UnaryInfC{$\Pi_3$}
\noLine
\UnaryInfC{$C$}
\RightLabel{\scriptsize{$\vee$-elim}}
\TrinaryInfC{$C$}
\end{prooftree}

\begin{prooftree}
\AxiomC{$\Gamma_1$}
\noLine
\UnaryInfC{$\Pi_1$}
\noLine
\UnaryInfC{$A$}
\AxiomC{$\Gamma_2$}
\noLine
\UnaryInfC{$\Pi_2$}
\noLine
\UnaryInfC{$B$}
\RightLabel{\scriptsize{$\wedge$-int}}
\BinaryInfC{$A \wedge B$}
\DisplayProof
\qquad
\AxiomC{$\Gamma$}
\noLine
\UnaryInfC{$\Pi$}
\noLine
\UnaryInfC{$A_1\wedge A_2$}
\RightLabel{\scriptsize{$\wedge_j$-elim}}
\UnaryInfC{$A_j$}
\end{prooftree}

\begin{prooftree}
\AxiomC{$\Gamma$}
\noLine
\UnaryInfC{$\Pi$}
\noLine
\UnaryInfC{$\bot$}
\RightLabel{\scriptsize{$\bot$-elim}}
\UnaryInfC{$A$}
\DisplayProof
\qquad
\AxiomC{$[\neg A^i]$}
\noLine
\UnaryInfC{$\Pi$}
\noLine
\UnaryInfC{$\bot$}
\RightLabel{\scriptsize{$A^c$-int}}
\UnaryInfC{$A^c$}
\DisplayProof
\qquad
\AxiomC{$\Gamma_1$}
%\noLine
%\UnaryInfC{}
%\noLine
%\UnaryInfC{}
%\noLine
%\UnaryInfC{}
%\noLine
%\UnaryInfC{}
\noLine
\UnaryInfC{$\Pi_1$}
\noLine
\UnaryInfC{$A^c$}
\AxiomC{$\Gamma_2$}
\noLine
%\UnaryInfC{}
%\noLine
%\UnaryInfC{}
%\noLine
%\UnaryInfC{}
%\noLine
%\UnaryInfC{}
%\noLine
\UnaryInfC{$\Pi_2$}
\noLine
\UnaryInfC{$\neg A^i$}
\RightLabel{\scriptsize{$A^c$ - elim}}
\BinaryInfC{$\bot$}
\end{prooftree}

\caption{Ecumenical natural deduction system $\NE$.}\label{fig:NEB}
\end{figure}

\subsection{Soundness}

Contrary to what happens with completeness, the proof of soundness follows easily from the proof in \cite{Sandqvist2015IL}.

\begin{lemma}
    $\vDash_{S} A^c$ iff $\vDash_{S} \neg \neg A^i$.
\end{lemma}

\begin{proof}
    ($\Rightarrow$) Suppose $\vDash_{S} A^c$. Then, for all $S \subseteq S'$, we have $A^i \nvDash_{S} \bot$. Suppose that, for some of those $S'$, $\vDash_{S'} \neg A^i$ Then we have $A^i \vDash_{S'} \bot$, which yields a contradiction. Thus for all such $S'$ we have $\nvDash_{S'} \neg A^i$, hence $\neg A^{i} \vDash_{S} \bot$ is vacuously satisfied and $\vDash_{S} \neg \neg A^i$ holds.

  ($\Leftarrow$) Suppose $\vDash_{S} \neg \neg A^i$. Then, for all $S \subseteq S'$, $\neg A^i \vDash_{S'} \bot$. Suppose that for some of those $S'$ it holds that $A^i \vDash_{S'} \bot$. Then we have $\vDash_{S'} \neg A^i$, and thus $\vDash_{S'} \bot$ and $\vdash_{S'} \bot$. Contradiction. Hence $A^i \nvDash_{S'} \bot$ for all $S \subseteq S^n$, and thus $\vDash_{S} A^c$.  
\end{proof}

\begin{theorem}[Soundness]\label{thm:soundness}
If $\Gamma\vdash_{\NE} A$ then $\Gamma \vDash A$.
\end{theorem}
\begin{proof}
Due to the collapse between local and global consequence in strong semantics, if we eliminate all clauses for classical formulas and define $A^c = \neg \neg A^i$ we get an equivalent definition. Then, since all the remaining semantic clauses are just Sandqvist's clauses for intuitionistic logic, our proof of soundness follows from his (provided $A^c$ is treated as $\neg \neg A^i$ on induction steps). The only important difference is in the treatment of $\bot - elim$, which is slightly different due to the consistency requirement. The induction hypothesis gives us $\Gamma \vDash \bot$ and thus $\Gamma \vDash_S \bot$ for arbitrary $S$, from which we conclude that for no $S \subseteq S'$ we have $\vDash_S B$ for all $B \in \Gamma$. But then $\Gamma \vDash_{S'} A$ holds vacuously for all such $S'$, which shows $\Gamma \vDash_S A$ for arbitrary $A$ and arbitrary $S$ and thus $\Gamma \vDash A$.
\end{proof}

%\cyan{VN: Mudei o sistema, nesse novo eu tenho certeza que d‡ pra provar completude com rela‹o ˆ no‹o forte! Vou terminar a prova depois.}

\subsection{Completeness}

We will now prove completeness for strong ecumenical semantics with respect to the natural deduction system shown in Fig. 1.  We use an adaptation of Sandqvist's proof \cite{Sandqvist2015IL}; changes are made only to deal with classical formulas and the consistency constraint.

\begin{lemma}\label{lemma:syntacticextensionexistence}
    $p \nvdash_S \bot$ iff the system $S'$ obtained by adding a rule which concludes $p$ from empty premises to a consistent $S$ is also consistent.
\end{lemma}

\begin{proof} This is strengthened syntactic counterpart of Lemma \ref{lemma:p^c implies existence of extension}.

      Assume $p \nvdash_{S} \bot$. Let $S'$ be the system obtained by adding a rule concluding $p$ from empty premises to $S'$. Suppose that $\vdash_{S'} \bot$. Then there is a deduction $\Pi$ in $S'$ showing $\bot$. If it does not use the new rule added to $S'$, $\Pi$ is also a deduction in $S$, so $S$ would violate the consistency requirement. If it does use the new rule, by replacing every occurrence of it by an assumption with shape $p$ we get a deduction showing $p \vdash_S \bot$, which contradicts our initial hypothesis. Since a contradiction is obtained in both cases, we conclude $\nvdash_{S'} \bot$

      For the other direction, assume the system $S'$ obtained by adding a rule which concludes $p$ from empty premises to $S$ is consistent. Assume $p \vdash_S \bot$. Since $S \subseteq S'$ we have $\vdash_{S'} \bot$, violating the consistency requirement. Thus, $p \nvdash_S \bot$.  
\end{proof}

%\begin{theorem}
 %   (\textbf{Completeness}) $\Gamma \vDash A$ implies $\Gamma \vdash_{\NE} A$.
%\end{theorem}

    Let $\Gamma^{Sub}$ be the set of all subformulas of formulas contained in a set $\Gamma$. Let $\Delta^c_\Gamma = \{ \neg A^i | A^c \in \Gamma \}$. Now, let $\Gamma^{\star} = ((\Gamma \cup \{ A\})^{Sub}) \cup (\Delta^c_{(\Gamma \cup \{A\})^{Sub})}))$
    
    We start by producing some mapping $\alpha$ which assigns to each formula $A$ in $\Gamma^{\star}$ a unique $p^A$ such that:

    \begin{enumerate}
        \item $p^{A}  = q$, if $A = q^i$ (for $q \in \Atb$);

        \medskip
        \item Else, $p^A \in \At$ and $(p^A)^i \notin \Gamma^{\star}$.
    \end{enumerate}

 Notice that, since the assigned atoms are unique, $p^A = p^B$ iff $A = B$.

    Consider now any semantic consequence $\Gamma \vDash A$. Fix any mapping $\alpha$ for $\Gamma^{\star}$. Following Sandqvist's strategy, we start by using the mapping $\alpha$ to build an atomic system $\mathcal{N}$ which is finely tailored for our proof.

    We start by defining atomic correspondents of the natural deduction rules (for $i \in \{1,2\}$):

    \begin{tabular}{lc@{\quad}l}
\end{tabular}
\begin{prooftree}
\AxiomC{$\Gamma\;[p^{A}]$}
\noLine
\UnaryInfC{$\Pi$}
\noLine
\UnaryInfC{$p^{B}$}
\RightLabel{\scriptsize{$p^{A \to B} - int$}}
\UnaryInfC{$p^{A \to B}$}
\DisplayProof
\qquad
\AxiomC{$\Gamma_1$}
\noLine
\UnaryInfC{$\Pi_1$}
\noLine
\UnaryInfC{$p^{A\to B}$}
\AxiomC{$\Gamma_2$}
\noLine
\UnaryInfC{$\Pi_2$}
\noLine
\UnaryInfC{$p^{A}$}
\RightLabel{\scriptsize{$p^{A \to B} - elim$}}
\BinaryInfC{$p^{B}$}
\end{prooftree}

\begin{prooftree}
\AxiomC{$\Gamma$}
\noLine
\UnaryInfC{$\Pi$}
\noLine
\UnaryInfC{$p^{A_i}$}
\RightLabel{\scriptsize{$p^{A_1 \vee A_2} - int$}}
\UnaryInfC{$p^{A_1\vee A_2}$}
\DisplayProof
\quad
\AxiomC{$\Gamma_1$}
\noLine
\UnaryInfC{$\Pi_1$}
\noLine
\UnaryInfC{$p^{A\vee B}$}
\AxiomC{$\Gamma_2\;[p^{A}]$}
\noLine
\UnaryInfC{$\Pi_2$}
\noLine
\UnaryInfC{$q$}
\AxiomC{$\Gamma_3\;[p^{B}]$}
\noLine
\UnaryInfC{$\Pi_3$}
\noLine
\UnaryInfC{$q$}
\RightLabel{\scriptsize{$p^{A \vee B}, q - elim$}}
\TrinaryInfC{$q$}
\end{prooftree}

\begin{prooftree}
\AxiomC{$\Gamma_1$}
\noLine
\UnaryInfC{$\Pi_1$}
\noLine
\UnaryInfC{$p^{A}$}
\AxiomC{$\Gamma_2$}
\noLine
\UnaryInfC{$\Pi_2$}
\noLine
\UnaryInfC{$p^{B}$}
\RightLabel{\scriptsize{$p^{A \wedge B} - int$}}
\BinaryInfC{$p^{A \wedge B}$}
\DisplayProof
\qquad
\AxiomC{$\Gamma$}
\noLine
\UnaryInfC{$\Pi$}
\noLine
\UnaryInfC{$p^{A_1\wedge A_2}$}
\RightLabel{\scriptsize{$p^{A_1 \wedge_i A_2} - elim$}}
\UnaryInfC{$p^{A_i}$}
\end{prooftree}

\begin{prooftree}
\AxiomC{$[p^{\neg A^i}]$}
\noLine
\UnaryInfC{$\Pi$}
\noLine
\UnaryInfC{$\bot$}
\RightLabel{\scriptsize{$p^{A^c} - int$}}
\UnaryInfC{$p^{A^c}$}
\DisplayProof
\qquad
\AxiomC{$\Gamma_1$}
\noLine
\UnaryInfC{$\Pi_1$}
%\noLine
%\UnaryInfC{}
%\noLine
%\UnaryInfC{}
%\noLine
%\UnaryInfC{}
%\noLine
%\UnaryInfC{}
\noLine
\UnaryInfC{$p^{A^c}$}
\AxiomC{$\Gamma_2$}
\noLine
\UnaryInfC{$\Pi_2$}
\noLine
%\UnaryInfC{}
%\noLine
%\UnaryInfC{}
%\noLine
%\UnaryInfC{}
%\noLine
%\UnaryInfC{}
%\noLine
\UnaryInfC{$p^{\neg A^i}$}
\RightLabel{\scriptsize{$p^{A^c}- elim$}}
\BinaryInfC{$\bot$}
\DisplayProof
\qquad
\AxiomC{$\Gamma$}
\noLine
\UnaryInfC{$\Pi$}
\noLine
\UnaryInfC{$\bot$}
\RightLabel{\scriptsize{$\bot$, $q - elim$}}
\UnaryInfC{$q$}
\end{prooftree}

We now add atomic rules to $\mathcal{N}$ for all formulas $D \in \Gamma^{\star}$, according to the following criteria.

\begin{enumerate}
    \item For every formula $D$ with shape $A \to B$ or $A \wedge B$, we add ${p^D}-int$ and ${p^D}-elim$ to $\mathcal{N}$.

    \medskip

    \item For every formula $D$ with shape $A \lor B$ we add the rules ${p^D}-int$ to $\mathcal{N}$, and for every $D$ with shape $A \lor B$ and every $q \in \Atb$ we add ${p^D, p}- elim$ to $\mathcal{N}$.

    \medskip

   % $p \in \At \cup \{ \bot \}$

    \item For every formula $D$ with shape $A^c$ we add ${p^D}- int$ and ${p^D}- elim$  to $\mathcal{N}$. Notice that, by the definition of $\Delta^c_{\Gamma}$ and $\Gamma^{\star}$, if $ A^c \in (\Gamma \cup \{ A\})^{Sub}$ then $\neg A^i \in \Gamma^*$;

    \medskip

    \item  For every $q \in \Atb$ we add $\bot , q - elim$ to $\mathcal{N}$;

    \medskip

    \item  We also stipulate that $\mathcal{N}$ contains no rules other than those added by this procedure.
    
\end{enumerate}

Since all atomic systems are now required to be consistent, before using $\mathcal{N}$ in the completeness proof we must prove that it is consistent. One interesting way to do this is by proving \textit{atomic normalization results} for $\mathcal{N}$.

We start by providing some definitions required for the normalization proof.

\begin{definition}
    Rules with shape ${p^{A \land B}}-int$, $p^{A_{1} \lor A_{2}} -int$, $p^{A \to B} - int$ and $p^{A^{c}} - int$ are introduction rules of $\mathcal{N}$. Rules with shape ${p^{A_{1} \land A_{2}}}- elim$, $p^{A \lor B}, q - elim$, $p^{A \to B} - elim$, $p^{A^{c}} - elim$ and $\bot - elim$ are elimination rules of $\mathcal{N}$.
\end{definition}

\begin{definition}
   For any rule $p^{D} - elim$ or $p^{D}, q - elim$ of $\mathcal{N}$, we say that the atom with shape $p^{D}$ occurring above it (or the leftmost occurrence if there is more than one) is the rule's major premise. The major premise of $\bot, q - elim$ is $\bot$. All other premises are the rule's minor premises.
\end{definition}

For the sake of simplicity, in some contexts we talk simply of $\land - int$, $\lor - int$, $\to - int$, $A^{c} - int$,  $\land - elim$, $\lor - elim$, $\to - elim$, $A^{c} - elim$ and $\bot-elim$ rules when referring to the rules of $\mathcal{N}$.

\begin{definition}
    The length of a derivation is the number of formula occurrences in it.
\end{definition}

\begin{definition}
The degree of a atomic formula $p^{A}$ relative to a previously fixed mapping $\alpha$ of formulas into atoms, denoted by $d[p^{A}]$, is recursively defined as:
\begin{enumerate} 
\item $d[p^{A}] = 0$, if $A = q^{i}$ for $q \in \Atb$;

\medskip

\item $d[p^{A^{c}}] = d[p^{A}] + 2$;

\medskip

\item $d[p^{\neg A}] = d[p^{A}] + 1$;

\medskip

\item $d[p^{A \wedge B}] = d[p^{A}] + d[p^{B}] + 1$;

\medskip

\item $d[p^{A \vee B}] = d[p^{A}] + d[p^{B}] + 1$;

\medskip

\item $d[p^{A \to B}] = d[p^{A}] + d[p^{B}] + 1$;
\end{enumerate}
\end{definition}

Notice that this is slightly different from Definition \ref{def:complexity} because now the degree of classical formulas needs to be the degree of its intuitionistic version plus $2$.

\begin{definition}
A formula occurrence in a derivation $\Pi$ in $\mathcal{N}$ that is at the same time the conclusion of an application of an introduction rule and the major premise of an elimination rule is said to be a \textit{maximum formula} in $\Pi$. 
\end{definition}

\begin{example}
The following are examples of maximum formula occurrences:
\begin{prooftree}
\AxiomC{$\Gamma_{1}$}
\noLine
\UnaryInfC{$\Pi_{1}$}
\noLine
\UnaryInfC{$p^{A}$}
\AxiomC{$[p^{A}]^{n}$}
\AxiomC{$\Gamma_{2}$}
\noLine
\BinaryInfC{$\Pi_{2}$}
\noLine
\UnaryInfC{$p^{B}$}
\RightLabel{\scriptsize{$p^{A \to B} -int, n$}}
\UnaryInfC{$p^{A \to B}$}
\RightLabel{\scriptsize{$p^{A \to B} -elim$}}
\BinaryInfC{$p^{B}$}
\end{prooftree}

\begin{prooftree}
\AxiomC{$[p^{\neg A^{i}}]^{n}$}
\noLine
\UnaryInfC{$\Pi_{1}$}
\noLine
\UnaryInfC{$\bot$}
\RightLabel{\scriptsize{$p^{A^{c}} - int, n$}}
\UnaryInfC{$p^{A^{c}}$}
\AxiomC{$\Pi_{2}$}
\noLine
\UnaryInfC{$p^{\neg A^{i}}$}
\RightLabel{\scriptsize{$p^{A^{c}} - elim$}}
\BinaryInfC{$\bot$}
\end{prooftree}
\end{example}

\begin{definition}
    A sequence $A_{1}, A_{2}, ... , A_{n}$ of formula occurrences in a deduction is a \textit{thread} if $A_{1}$ is a (possibly discharged) assumption, $A_{m}$ stands immediately below $A_{m-1}$ for every $1 < m \leq n$ and $A_{n}$ is the conclusion of the deduction.
\end{definition}

\begin{definition}
A sequence $A_{1}, A_{2}, ... , A_{n}$ of consecutive formula occurrences in a thread   is a \textit{segment} if and only if it satisfies the following conditions:
\begin{enumerate}
\item $A_{1}$ is not the consequence of an application of $\vee$-elimination;

\medskip

\item Each $A_{m}$, for $m < n$, is a minor premiss of an application of $\vee$-elimination;

\medskip

\item $A_{n}$ is not the minor premiss of an application of $\vee$-elimination.
\end{enumerate}
\end{definition}

\begin{definition}
The last formula of a segment is called the \textit{vertex}  of the segment.
\end{definition}

\begin{definition}
A segment that begins with an application of an introduction rule or the $\bot - elim$ rule and ends with the major premise of an elimination rule is said to be a \textit{maximum segment}.
\end{definition}

Example:
\begin{prooftree}
\AxiomC{$\Gamma_{1}$}
\noLine
\UnaryInfC{$\Pi_{1}$}
\noLine
\UnaryInfC{$p^{A \vee B}$}
\AxiomC{$[p^{A}]^{n}$}
\AxiomC{$\Gamma_{2}$}
\noLine
\BinaryInfC{$\Pi_{2}$}
\noLine
\UnaryInfC{$p^{A_{1}}$}
\AxiomC{$[p^{A}]^n$}
\AxiomC{$\Gamma_{3}$}
\noLine
\BinaryInfC{$\Pi_{3}$}
\noLine
\UnaryInfC{$p^{A_{2}}$}
\RightLabel{\scriptsize{$p^{A_{1} \land A_{2}} - int$}}
\BinaryInfC{$p^{ A_{1} \wedge A_{2}}$}
\AxiomC{$[p^{B}]^{m}$}
\AxiomC{$\Gamma_{4}$}
\noLine
\BinaryInfC{$\Pi_{4}$}
\noLine
\UnaryInfC{$p^{A_{1} \land A_{2}}$}
\RightLabel{\scriptsize{$p^{A \lor B}, p^{A_{1} \land A_{2}}- elim, n, m$}}
\TrinaryInfC{$p^{A_{1} \land A_{2}}$}
\RightLabel{\scriptsize{$p^{A_{1} \land A_{2}} - elim$}}
\UnaryInfC{$p^{A_{1}}$}
\end{prooftree}

\medskip

The sequence $p^{A_{1} \wedge A_{2}}$, $p^{A_{1} \wedge A_{2}}$ starting with the application of $\wedge$-Introduction and ending with an application of $\wedge$-elimination is a \textit{maximum segment}\footnote{Maximum formulas are special cases of maximum segments.}.

\begin{definition}
The degree of a derivation $\Pi$ in $\mathcal{N}$, $d[\Pi ]$, is defined as $d[\Pi ] = max\{ d[p^{C}]: p^{C}$ is a maximum formula or the vertex of a maximum segment in $\Pi \}$.
\end{definition}

We adapt Prawitz's usual proper and permutative reductions for intuitionistic logic \cite{prawitz1965} to the system $\mathcal{N}$. Besides the usual reductions for the operators $\wedge$, $\to$, $\vee$ and $\neg$, we have a new reduction for maximum formulas of the form $A^{c}$:

\begin{prooftree}
\AxiomC{$[p^{\neg A^{i}}]^{n}$}
\noLine
\UnaryInfC{$\Pi_{1}$}
\noLine
\UnaryInfC{$\bot$}
\RightLabel{\scriptsize{$p^{A^{c}} - int, n$}}
\UnaryInfC{$p^{A^{c}}$}
\AxiomC{$\Pi_{2}$}
\noLine
\UnaryInfC{$p^{\neg A}$}
\RightLabel{\scriptsize{$p^{A^c}-elim$}}
\BinaryInfC{$\bot$}
\noLine
\UnaryInfC{$\Pi_{3}$}
\DisplayProof
\qquad
\AxiomC{reduces to}
\DisplayProof
\qquad
\AxiomC{$\Pi_{2}$}
\noLine
\UnaryInfC{$[p^{\neg A^{i}}]$}
\noLine
\UnaryInfC{$\Pi_{1}$}
\noLine
\UnaryInfC{$\bot$}
\noLine
\UnaryInfC{$\Pi_{3}$}
\end{prooftree}

\medskip

\begin{definition}
A derivation $\Pi$ of $\Gamma \vdash_{\mathcal{N}} A$ is called \textit{critical} if and only if:
\begin{enumerate}
\item $\Pi$ ends with an application $\beta$ of an elimination rule.

\medskip

\item The major premiss $p^{B}$ of $\beta$ is a maximum formula or the vertex of a maximum segment.

\medskip

\item $d[\Pi] = d[p^{B}]$

\medskip

\item Every other maximum formula or maximum segment  in $\Pi$ has a degree smaller than $d[p^{B}]$.
\end{enumerate}
\end{definition}

\begin{lemma}
Let $\Pi$ be a critical derivation of $A$ from $\Gamma$. Then $\Pi$ reduces to a derivation $\Pi'$ of $A$ from $\Gamma$ such that $d[\Pi' ] < d[\Pi ]$
\end{lemma}

\begin{proof}
    Induction over the length of $\Pi$. There are two cases to be examined depending on whether $d[\Pi ]$ is determined by a maximum formula or by  the vertex of a maximum segment.\\

\noindent Case 1: $d[\Pi ]$ is determined by a maximum formula. The result follows directly from the application of a reduction to this maximum formula. %Let us examine three cases.
\begin{enumerate}

\medskip

\item The critical derivation $\Pi$ is:

\begin{prooftree}
\AxiomC{$\Pi_{1}$}
\noLine
\UnaryInfC{$p^{A_{1}}$}
\AxiomC{$\Pi_{2}$}
\noLine
\UnaryInfC{$p^{A_{2}}$}
\RightLabel{\scriptsize{$p^{A_{1} \land A_{2}} - int$}}
\BinaryInfC{$p^{A_{1} \wedge A_{2}}$}
\RightLabel{\scriptsize{$p^{A_{1} \land A_{2}} - elim$}}
\UnaryInfC{$p^{A_{i}}$}
\end{prooftree}

We know that $d(\Pi ) = d[p^{A_{1} \wedge A_{2}}] > d[\Pi_{i} ]$ (for $i \in \{1,2\}$). $\Pi$ reduces to

\begin{prooftree}
\AxiomC{$\Pi_{i}$}
\noLine
\UnaryInfC{$p^{A_{i}}$}
\end{prooftree} 
And the degree of this derivation is equal to $d[\Pi_{i} ]$ which is smaller than $d[\Pi ]$.

\medskip
\item The critical derivation $\Pi$ is:
\begin{prooftree}
\AxiomC{$\Gamma_{1}$}
\noLine
\UnaryInfC{$\Pi_{1}$}
\noLine
\UnaryInfC{$p^{A}$}
\AxiomC{$[p^{A}]^{n}$}
\AxiomC{$\Gamma_{2}$}
\noLine
\BinaryInfC{$\Pi_{2}$}
\noLine
\UnaryInfC{$p^{B}$}
\RightLabel{\scriptsize{$p^{A \to B} - int, n$}}
\UnaryInfC{$p^{A \to B}$}
\RightLabel{\scriptsize{$p^{A \to B} - elim$}}
\BinaryInfC{$p^{B}$}
\end{prooftree}

\medskip

$\Pi$ reduces to the following derivation $\Pi'$:

\begin{prooftree}
\AxiomC{$\Gamma_{1}$}
\noLine
\UnaryInfC{$\Pi_{1}$}
\noLine
\UnaryInfC{$[p^{A}]$}
\AxiomC{$\Gamma_{2}$}
\noLine
\BinaryInfC{$\Pi_{2}$}
\noLine
\UnaryInfC{$p^{B}$}
\end{prooftree}

\medskip

We can easily see that $d[\Pi' ] \leq max\{ d[\Pi_{1} ] , d[\Pi_{2} ], d[p^{A}] \} < d[\Pi ] = d[p^{A \to B}]$.

\medskip

\item The critical derivation $\Pi$ is:

\begin{prooftree}
\AxiomC{$\Gamma$}
\noLine
\UnaryInfC{$\Pi$}
\noLine
\UnaryInfC{$p^{A_{i}}$}
\RightLabel{\scriptsize{$p^{A_{1} \lor A_{2}} - int$}}
\UnaryInfC{$p^{A_{1} \lor A_{2}}$}
\AxiomC{$[p^{A_{1}}]^{n}$}
\AxiomC{$\Gamma_{1}$}
\noLine
\BinaryInfC{$\Pi_{1}$}
\noLine
\UnaryInfC{$q$}
\AxiomC{$[p^{A_{2}}]^{m}$}
\AxiomC{$\Gamma_{2}$}
\noLine
\BinaryInfC{$\Pi_{2}$}
\noLine
\UnaryInfC{$q$}
\RightLabel{\scriptsize{$p^{A_{1} \lor A_{2}}, q - elim, n, m$}}
\TrinaryInfC{$q$}
\end{prooftree}

\medskip

$\Pi$ reduces to the following derivation $\Pi'$:

\begin{prooftree}
\AxiomC{$\Gamma$}
\noLine
\UnaryInfC{$\Pi$}
\noLine
\UnaryInfC{$[p^{A_{i}}]$}
\AxiomC{$\Gamma_{i}$}
\noLine
\BinaryInfC{$\Pi_{i}$}
\noLine
\UnaryInfC{$q$}
\end{prooftree}

\medskip

We can easily see that $d[\Pi' ] \leq max\{ d[\Pi_{1} ] , d[\Pi_{2} ], d[\Pi], d[p^{A_{i}}] \} < d[\Pi ] = d[p^{A_{1} \lor A_{2}}]$.

\medskip

\item The critical derivation $\Pi$ is:

\begin{prooftree}
\AxiomC{$[p^{\neg A^{i}}]^{n}$}
\noLine
\UnaryInfC{$\Pi_{1}$}
\noLine
\UnaryInfC{$\bot$}
\RightLabel{\scriptsize{$p^{A^{c}}-int, n$}}
\UnaryInfC{$p^{A^{c}}$}
\AxiomC{$\Pi_{2}$}
\noLine
\UnaryInfC{$p^{\neg A^{i}}$}
\RightLabel{\scriptsize{$p^{A^{c}}-elim$}}
\BinaryInfC{$\bot$}
\end{prooftree}

\medskip

$\Pi$ reduces to the following derivation $\Pi'$:

\begin{prooftree}
\AxiomC{$\Pi_{2}$}
\noLine
\UnaryInfC{$[p^{\neg A^{i}}]$}
\noLine
\UnaryInfC{$\Pi_{1}$}
\noLine
\UnaryInfC{$\bot$}
\end{prooftree}

\medskip

We can easily see that $d[\Pi' ] \leq max\{ d[\Pi_{1} ] , d[\Pi_{2} ], d[p^{\neg A}] \} < d[\Pi ] = d[p^{A^{c}}]$\footnote{This is the step in which we cannot use Definition \ref{def:complexity}.}.

\medskip

\item The critical derivation was obtained through an application of $\bot - elim$. Then:

\begin{prooftree}
\AxiomC{$\Pi_{1}$}
\noLine
\UnaryInfC{$\bot$}
\RightLabel{\scriptsize{$\bot, p^{A_{1} \land A_{2}} - elim$}}
\UnaryInfC{$p^{A_{1} \land A_{2}}$}
\RightLabel{\scriptsize{$p^{A_{1} \land A_{2}}- elim$}}
\UnaryInfC{$p^{A_{i}}$}
\DisplayProof
\qquad
reduces to
\qquad
\AxiomC{$\Pi_{1}$}
\noLine
\UnaryInfC{$\bot$}
\RightLabel{\scriptsize{$\bot, p^{A_{i}} - elim$}}
\UnaryInfC{$p^{A_{i}}$}
\end{prooftree}

\bigskip

\begin{bprooftree}\hspace{-95pt}
\AxiomC{$\Pi_{1}$}
\noLine
\UnaryInfC{$\bot$}
\RightLabel{\scriptsize{$\bot, p^{A \lor B} - elim$}}
\UnaryInfC{$p^{A \lor B}$}
\AxiomC{$[p^{A}]^{n}$}
\noLine
\UnaryInfC{$\Pi_{2}$}
\noLine
\UnaryInfC{$q$}
\AxiomC{$[p^{B}]^{m}$}
\noLine
\UnaryInfC{$\Pi_{3}$}
\noLine
\UnaryInfC{$q$}
\RightLabel{\scriptsize{$p^{A \lor B}, q - elim$}}
\TrinaryInfC{$q$}
\qquad
\end{bprooftree}
reduces to
\begin{bprooftree}
\qquad
\AxiomC{$\Pi_{1}$}
\noLine
\UnaryInfC{$\bot$}
\RightLabel{\scriptsize{$\bot, q- elim$}}
\UnaryInfC{$q$}
\end{bprooftree}

\bigskip

\begin{bprooftree}
\hspace{-50pt}
\AxiomC{$\Pi_{1}$}
\noLine
\UnaryInfC{$\bot$}
\RightLabel{\scriptsize{$\bot, p^{A \to B} - elim$}}
\UnaryInfC{$p^{A \to B}$}
\AxiomC{$\Pi_{2}$}
\noLine
\UnaryInfC{$p^{A}$}
\RightLabel{\scriptsize{$p^{A \to B }- elim$}}
\BinaryInfC{$p^{B}$}
\end{bprooftree}
\qquad
reduces to
\qquad
\begin{bprooftree}
\AxiomC{$\Pi_{1}$}
\noLine
\UnaryInfC{$\bot$}
\RightLabel{\scriptsize{$\bot, p^{B} - elim$}}
\UnaryInfC{$p^{B}$}
\end{bprooftree}

\bigskip

\begin{bprooftree}
\hspace{-30pt}
\AxiomC{$\Pi_{1}$}
\noLine
\UnaryInfC{$\bot$}
\RightLabel{\scriptsize{$\bot, p^{A^{c}} - elim$}}
\UnaryInfC{$p^{A^{c}}$}
\AxiomC{$\Pi_{2}$}
\noLine
\UnaryInfC{$p^{\neg A^{i}}$}
\RightLabel{\scriptsize{$p^{A^{c}}- elim$}}
\BinaryInfC{$\bot$}
\end{bprooftree}
\qquad
reduces to
\qquad
\begin{bprooftree}
\AxiomC{$\Pi_{1}$}
\noLine
\UnaryInfC{$\bot$}
\end{bprooftree}

\bigskip

In all cases it is straightforward to check that the degree of the derivation is reduced.

\end{enumerate}

%%%HANDLE

Case 2: $d[\Pi ]$ is determined by the vertex of a maximal segment. $\Pi$ is:

\begin{prooftree}
\AxiomC{$\Gamma_{1}$}
\noLine
\UnaryInfC{$\Pi_{1}$}
\noLine
\UnaryInfC{$p^{A \vee B}$}
\AxiomC{$[p^{A}]^{n}$}
\AxiomC{$\Gamma_{2}$}
\noLine
\BinaryInfC{$\Pi_{2}$}
\noLine
\UnaryInfC{$q$}
\AxiomC{$[p^{B}]^{m}$}
\AxiomC{$\Gamma_{3}$}
\noLine
\BinaryInfC{$\Pi_{3}$}
\noLine
\UnaryInfC{$q$}
\RightLabel{\scriptsize{$p^{A \lor B}, q - elim, n, m$}}
\TrinaryInfC{$q$}
\AxiomC{$\Sigma_{1}$}
\noLine
\UnaryInfC{$\Delta_{1}$}
\noLine
\UnaryInfC{$r_{1}$}
\AxiomC{$...$}
\AxiomC{$\Sigma_{k}$}
\noLine
\UnaryInfC{$\Delta_{k}$}
\noLine
\UnaryInfC{$r_{k}$}
\QuaternaryInfC{$r$}

\end{prooftree}

\medskip

By means of a permutative reduction $\Pi$ reduces to the following derivation $\Pi^{*}$:

\begin{prooftree}
\AxiomC{$\Gamma_{1}$}
\noLine
\UnaryInfC{$\Pi_{1}$}
\noLine
\UnaryInfC{$p^{A \vee B}$}
\AxiomC{$[p^{A}]^{n}$}
\AxiomC{$\Gamma_{2}$}
\noLine
\BinaryInfC{$\Pi_{2}$}
\noLine
\UnaryInfC{$q$}
\AxiomC{$\Sigma_{1}$}
\noLine
\UnaryInfC{$\Delta_{1}$}
\noLine
\UnaryInfC{$r_{1}$}
\AxiomC{$...$}
\AxiomC{$\Sigma_{k}$}
\noLine
\UnaryInfC{$\Delta_{k}$}
\noLine
\UnaryInfC{$r_{k}$}
\QuaternaryInfC{$r$}
\AxiomC{$[p^{B}]^{m}$}
\AxiomC{$\Gamma_{3}$}
\noLine
\BinaryInfC{$\Pi_{3}$}
\noLine
\UnaryInfC{$q$}
\AxiomC{$\Sigma_{1}$}
\noLine
\UnaryInfC{$\Delta_{1}$}
\noLine
\UnaryInfC{$r_{1}$}
\AxiomC{$...$}
\AxiomC{$\Sigma_{k}$}
\noLine
\UnaryInfC{$\Delta_{k}$}
\noLine
\UnaryInfC{$r_{k}$}
\QuaternaryInfC{$r$}
\RightLabel{\scriptsize{$p^{A \lor B}, r - elim, n, m$}}
\TrinaryInfC{$r$}
\end{prooftree}

Without loss of generality, we can assume that the two derivations of the minor premises of the application of $\vee$-elimination are critical. By the induction hypothesis, they reduce to derivations

\begin{prooftree}
\AxiomC{$p^{A}$}
\AxiomC{$\Gamma_{2}'$}
\noLine\
\BinaryInfC{$\Pi_{2}'$}
\noLine
\UnaryInfC{$r$}
\DisplayProof
\qquad
and
\qquad
\AxiomC{$p^{B}$}
\AxiomC{$\Gamma_{3}'$}
\noLine\
\BinaryInfC{$\Pi_{3}'$}
\noLine
\UnaryInfC{$r$}
\end{prooftree}

\medskip

such that $d[\Pi_{2}'] < d[\Pi ]$ and $d[\Pi_{3}'] < d[\Pi ]$. We can then take $\Pi'$ to be:

\begin{prooftree}
\AxiomC{$\Gamma_{1}$}
\noLine
\UnaryInfC{$\Pi_{1}$}
\noLine
\UnaryInfC{$p^{A \vee B}$}
\AxiomC{$[p^{A}]^{n}$}
\AxiomC{$\Gamma_{2}'$}
\noLine\
\BinaryInfC{$\Pi_{2}'$}
\noLine
\UnaryInfC{$r$}
\AxiomC{$[p^{B}]^{m}$}
\AxiomC{$\Gamma_{3}'$}
\noLine\
\BinaryInfC{$\Pi_{3}'$}
\noLine
\UnaryInfC{$r$}
\RightLabel{\scriptsize{$p^{A \lor B}, r - elim, n, m$}}
\TrinaryInfC{$r$}
\end{prooftree}  

\medskip

\end{proof}

\begin{lemma}\label{lemma:mainlemmanormalization}
Let $\Pi$ be a derivation of $A$ from $\Gamma$ in $\mathcal{N}$ such that $d[\Pi ] > 0$. Then $\Pi$ reduces to a derivation $\Pi'$ of $A$ from $\Gamma$ in $\mathcal{N}$ such that $d[\Pi' ] < d[\Pi ]$
\end{lemma}
Proof: By induction over the length of $\Pi$. We examine two cases depending on the form of the last rule applied in $\Pi$.
\begin{enumerate}
\item The last rule applied in $\Pi$ is and introduction rule. The result follows directly from the induction hypothesis.
\item The last rule applied in $\Pi$ is an elimination rule. $\Pi$ has the following general form:

\begin{prooftree}
\AxiomC{$\Pi_{1}$}
\noLine
\UnaryInfC{$p_{1}$}
\AxiomC{$...$}
\AxiomC{$\Pi_{n}$}
\noLine
\UnaryInfC{$p_{n}$}
\TrinaryInfC{$p$}
\end{prooftree}

By the induction hypothesis, each derivation 
\begin{prooftree}
\AxiomC{$\Pi_{i}$}
\noLine
\UnaryInfC{$p_{i}$}
\end{prooftree}
$(1 \leq i \leq n)$ reduces to a derivation
\begin{prooftree}
\AxiomC{$\Pi'_{i}$}
\noLine
\UnaryInfC{$p_{i}$}
\end{prooftree}
such that $d[\Pi'_{i} ] < d[\Pi_{i} ]$.  Let $\Pi^{*}$ be:

\begin{prooftree}
\AxiomC{$\Pi'_{1}$}
\noLine
\UnaryInfC{$p_{1}$}
\AxiomC{$...$}
\AxiomC{$\Pi'_{n}$}
\noLine
\UnaryInfC{$p_{n}$}
\TrinaryInfC{$p$}
\end{prooftree}

If $d[\Pi^{*} ] < d[\Pi ]$, we can take $\Pi' = \Pi^{*}$. If $d[\Pi^{*} ] = d[\Pi ]$, then $\Pi^{*}$ is a critical derivation and, by Lemma \ref{lemma:mainlemmanormalization}, it reduces to a derivation $\Pi'$ such that $d[\Pi' ] < d[\Pi ]$.  

\end{enumerate}

\begin{theorem}\label{thm:normalization}
Let $\Pi$ be a derivation of $A$ from $\Gamma$ in $\mathcal{N}$. Then $\Pi$ reduces to a normal derivation $\Pi'$ of $A$ from $\Gamma$ in $\mathcal{N}$.
\end{theorem}
\begin{proof}
     Directly from Lemma \ref{lemma:mainlemmanormalization} by induction on $d[\Pi ]$.
\end{proof}

Our choice of normalization strategy is purely incidental; we could as well have used Prawitz's original strategy or any other. What matters is that we have shown that $\mathcal{N}$ satisfies normalization, which we can now use to prove its consistency.

\begin{definition}
    Let $\Pi$ be a derivation in $\mathcal{N}$ and $A$ any formula occurrence in $\Pi$. The derivation $\Pi'$ obtained by removing from $\Pi$ all formula occurrences except $A$ and those above $A$ is called a subderivation of $\Pi$.
\end{definition}

\begin{lemma}\label{lemma:subderivationofnormalisnormal}
    If $\Pi$ is a normal derivation in $\mathcal{N}$ then all its subderivations are also normal.
\end{lemma}

\begin{proof}
    Let $\Pi$ be a normal derivation and $\Pi'$ any of its subderivations. It is straightforward to see that if $\Pi'$ contains a maximal formula or segment then that formula or segment is also maximal in $\Pi$, contradicting the assumption that $\Pi$ was normal. Therefore, no subderivation $\Pi'$ of $\Pi$ contains a maximal formula or segment, so every such $\Pi'$ is normal.   
\end{proof}

\begin{lemma}\label{lemma:alwaysundischarged}
If $\Pi$ is a normal derivation in $\mathcal{N}$ that does not end with an application of a introduction rule, then $\Pi$ contains at least one undischarged assumption.
\end{lemma}

\begin{proof}
    We prove the result by induction on the length of derivations.

\begin{enumerate}
    \item Base case: $\Pi$ has length $1$. Then the derivation is just a single occurrence of an assumption $p$ and shows $p \vdash_{\mathcal{N}} p$, so it depends on the undischarged assumption $p$.

    \medskip

\item $\Pi$ has length greater than $1$ and ends with an application of a elimination rule. Then consider the subderivation $\Pi'$ of $\Pi$ which has as its conclusion the major premise of the last rule applied in $\Pi$. Since $\Pi'$ is a subderivation of $\Pi$ and $\Pi$ is normal, by Lemma \ref{lemma:subderivationofnormalisnormal} we have that $\Pi'$ is normal. Notice that, if $\Pi'$ ended with an application of a introduction rule, since its conclusion is the major premise of an elimination rule there would be a maximum formula in $\Pi$, so since $\Pi$ is normal $\Pi'$ cannot end with an introduction rule. But then $\Pi'$ is a deduction with length smaller than that of $\Pi$ that does not end with an introduction rule, hence by the induction hypothesis it has at least one undischarged assumption. Therefore, since no elimination rule is capable of discharging assumptions occurring above its major premise, we conclude that the open assumption of $\Pi'$ are not discharged by the last rule application and so are also open assumptions of $\Pi$. 
\end{enumerate}
\end{proof}

%%%CONSISTENT CONTINUAR HANDLE

\begin{theorem}\label{thm:atomicconsistencyproof}
     $\mathcal{N}$ is consistent.
\end{theorem}

\begin{proof}
Assume, for the sake of contradiction, that there is a derivation $\Pi$ showing $\vdash_{\mathcal{N}} \bot$. By Theorem \ref{thm:normalization}, $\Pi$ reduces to a normal derivation $\Pi'$ showing  $\vdash_{\mathcal{N}} \bot$. A quick inspection of the shape of introduction rules reveals that no introduction rule can have $p^{\bot^{i}} (\equiv \bot)$ as its conclusion, hence the last rule of $\Pi'$ cannot be an introduction rule. But then from Lemma \ref{lemma:alwaysundischarged} it follows that $\Pi'$ must have at least one undischarged assumption, hence it cannot be a derivation showing $\vdash_{\mathcal{N}} \bot$. Contradiction. Therefore, $\nvdash_{\mathcal{N}} \bot$.    
\end{proof}

Now we have assured ourselves that $\mathcal{N}$ is indeed a valid atomic system, so we can proceed with the completeness proof.

%\cyan{VN: Acho que tem um pequeno erro na prova dos Lemmas 5 e 6, preciso dar uma revisada depois. O ideal Ž trocar coisas como ($\vdash p^a $ implies $\vdash p^a $) diretamente por $p^a \vdash p^B$}.

\begin{lemma}\label{lemma:prepcomplete}
 For all $A \in \Gamma^{\star}$ and all $\mathcal{N} \subseteq S$ it holds that $\vDash_{S} A$ iff $\vdash_S p^{A}$.   
\end{lemma}

%\begin{lemma}
 %   $\vdash_\mathcal{N} p^{A \vee B}$ iff, for every $\mathcal{N} \subseteq S$,  $p^A \vdash_{S} p$ and $p^B \vdash_{S} p$ implies $\vdash_{S} p$, for any $p$.
%\end{lemma}

\begin{proof}
We show the result by induction on the degree of formulas.

%In every case except that of the last one we prove the result for the local consequence relation $\Vdash_S^{G} A$

\begin{enumerate}
    \item $A = p^i$, for some $p \in \Atb$. Then $p^A = A$, and the result follows immediately from Clause 1 of strong validity;

    \medskip
    
    \item  $A = A \land B$.

 \medskip
 
 ($\Rightarrow$) Assume $\vDash_S A \land B$. Then $\vDash_S A$ and $\vDash_S B$. Induction hypothesis: $\vdash_S p^{A}$ and $\vdash_S p^{B}$. By $p^{A \land B} - int$, we obtain $\vdash_S p^{A \land B}$.

\medskip

 ($\Leftarrow$) Assume $\vdash_{S} p^{A \land B}$. Then, by $p^{A \land B} - elim$  we get both $\vdash_{S} p^{A}$ and $\vdash_{S} p^{B}$. Induction hypothesis: $\vDash_S A$ and $\vDash_S B$. Then, by the semantic clause for conjunction,  $\vDash_S A \land B$.

\medskip

    \item  $A = A \lor B$.

 \medskip
 
 ($\Rightarrow$) Assume $\vDash_S A \lor B$. Then, for every $S \subseteq S'$, $A \vDash_{S'} q^i$ and $B \vDash_{S'} q^i$ implies $\vDash_{S'} q^i$, for any $q \in \Atb$. Induction hypothesis: for all $S \subseteq S'$, $p^{A} \vdash_{S'} q$ and $p^{B} \vdash_{S'} q$ implies $\vdash_{S'} q$, for any $q \in \Atb$. Let $S' = S$. By the rules for $p^{A \lor B} - int$ we can conclude both $p^{A} \vdash_{S} p^{A \lor B}$ and $p^{B} \vdash_{S} p^{A \lor B}$, which we then use to conclude $\vdash_{S} p^{A \lor B}$.

\medskip

 ($\Leftarrow$) Assume $\vdash_{S} p^{A \lor B}$. Let $S'$ be any extension of $S$ such that $p^{A} \vdash_{S'} q$ and $p^{B} \vdash_{S'} q$, for some $q \in \Atb$. Clearly, since $S \subseteq S'$, $\vdash_{S'} p^{A \lor B}$.  Then, regardless of our choice of $q$, the rule $p^{A \lor B}, q - elim$ can be used to show $\vdash_{S'} q$. Induction hypothesis: for every $S \subseteq S'$, $A \vDash_{S'} q^i$ and $B \vDash_{S'} q^i$ implies $\vDash_{S'} q^i$, for any $q \in \Atb$. By the clause for disjunction we immediately conclude $\vDash_{S} A \lor B$.

\medskip

\item $A = A \to B$.

\medskip

 ($\Rightarrow$) Assume $\vDash_S A \to B$. Then, for any $S \subseteq S'$, $A \vDash_{S'} B$. Induction hypothesis: for any $S \subseteq S'$, $p^{A} \vdash_{S'} p^{B}$. Let $S' = S$. Then we can use the rule $p^{A \to B} - int$ to conclude $\vdash_{S} p^{A \to B}$.

\medskip

 ($\Leftarrow$) Assume $\vdash_S p^{A \to B}$. Let $S'$ be any extension of $S'$ such that $\vdash_{S'} p^{A}$. Since $S \subseteq S'$, we have  $\vdash_{S'} p^{A \to B}$. We can use the rule $p^{A \to B} - elim$ to conclude $\vdash_{S'} p^{B}$. Then we have that for any $S \subseteq S'$ it holds that $\vdash_{S'} p^{A}$ implies $\vdash_{S'} p^{B}$. Induction hypothesis: for any $S \subseteq S'$ it holds that $\vDash_{S'} p^{A}$ implies $\vDash_{S'} p^{B}$. By clause 7 of the definition of strong validity this yields $A \vDash_S B$, and thus $\vDash_S A \to B$.

 \medskip

\medskip

\item $A = A^c$.

\medskip

 ($\Rightarrow$) Assume $\vDash_S A^c$. Then, for any $S \subseteq S'$, $A^i \nvDash_{S^n} \bot$.
 
 Induction hypothesis: for any $S \subseteq S'$, $p^{A^i} \nvdash_{S^n} \bot$. Then, by Lemma \ref{lemma:syntacticextensionexistence}, for every $S'$ we have that the system obtained by adding a rule concluding $p^{A^i}$ from empty premises to $S'$ is consistent.

 Assume, for the sake of contradiction, that $p^{\neg A^i} \nvdash_S \bot$. Then, by Lemma \ref{lemma:syntacticextensionexistence}, the system $S'$ obtained by adding a rule concluding $p^{\neg A^i}$ from empty premises to $S$ is consistent. But by the previous result we also have that the system $S''$ obtained by adding a rule concluding $p^{A^i}$ from empty premises to $S'$ must be consistent. However, since $\vdash_{S''} p^{A^i}$ and $\vdash_{S''} p^{\neg A^i}$, we can apply the $p^{\neg A^i} - elim$ rule to show $\vdash_{S''} p^{\bot}$, and thus $\vdash_{S''} \bot$ due to the properties of the mapping $\alpha$. Contradiction. Thus, $p^{\neg A^i} \vdash_S \bot$, and so $ \vdash_S p^{A^c}$ can by obtained trough an application of $p^{A^c} - int$.

\medskip

 ($\Leftarrow$) Assume $\vdash_S p^{A^c}$. Suppose there is an $S \subseteq S'$ such that $p^{A} \vdash_{S'} p^{\bot}$. Then, by $p^{\neg A^i} - int$ we conclude $\vdash_{S'} p^{\neg A^i}$ and, since $S \subseteq S'$ and thus  $\vdash_{S'} p^{A^c}$, we conclude $\vdash_{S'} p^\bot$ through an application of $p^{A^c} - elim$, and thus $\vdash_{S'} \bot$. Contradiction. Hence, for all $S \subseteq S'$ we have $p^{A^i} \nvdash_{S'} \bot$. Induction hypothesis: for all $S \subseteq S'$ it holds that $A^i \nvDash_{S'} \bot$, which by the clauses for classical formulas yield $\vDash_{S} A^c$.  
  \end{enumerate}
\end{proof}

\begin{theorem}[Completeness] \label{thm:completeness}
$\Gamma \vDash A$ implies $\Gamma \vdash_{\NE} A$.
\end{theorem}

\begin{proof}

 Define a mapping $\alpha$ and a system $\mathcal{N}$ for $\Gamma$ and $A$ as shown earlier. Define a set $\Gamma^{\Atb} = \{p^{A} | A \in \Gamma \}$.

Suppose $\Gamma \vDash A$. By the definition of strong validity, we have $\Gamma \vDash_{\mathcal{N}} A$. Now define $\mathfrak{B}$ as the system obtained from $\mathcal{N}$ by adding a rule concluding $p^{B}$ from empty premises for every $p^{B} \in \Gamma^{\Atb}$. 

We split the the proof in two cases:

\begin{enumerate}
    \item $\mathfrak{B}$ is consistent. Then it is a valid extension of $\mathcal{N}$. By the definition of $\mathfrak{B}$, we have $\vdash_{\mathfrak{B}} p^{B}$ for all $p^{B} \in \Gamma^{\Atb}$. By Lemma \ref{lemma:prepcomplete}, for all $B \in \Gamma^\star$ we have that, for any $\mathcal{N} \subseteq S$, $\vdash_S p^{B}$ iff $\vDash_S B$. Since $\Gamma \subseteq \Gamma^{\star}$, we conclude $\vDash_{\mathfrak{B}} B$ for all $B \in \Gamma$. Since $\Gamma \vDash_{\mathcal{N}} A$ and $\mathcal{N} \subseteq \mathfrak{B}$, we also have $\vDash_{\mathfrak{B}} A$, and so by another application of Lemma \ref{lemma:prepcomplete} we conclude $\vdash_{\mathfrak{B}} p^{A}$. Thus, we conclude that there is a deduction $\Pi$ of $p^{A}$ in $\mathfrak{B}$.

    \medskip

    If the deduction does not use any of the rules contained in $\mathfrak{B}$ but not in $\mathcal{N}$, $\Pi$ is a deduction in $\mathcal{N}$, and so $\vDash_{\mathcal{N}} p^{A}$. If it does use some of the rules, by replacing every new rule concluding $p^{B}$ by an assumption $p^{B}$ we obtain a deduction $\Pi'$ of $\mathcal{N}$ which shows $\Delta \vdash_{\mathcal{N}} p^{A}$, for some $\Delta \subseteq \Gamma^{\Atb}$. Thus, in any case we obtain some deduction showing $\Gamma^{\At} \vdash_{\mathcal{N}} p^{A}$

    \medskip

    Let $\Pi$ be the deduction showing $\Gamma^{\At} \vdash_{\mathcal{N}} p^{A}$ obtained earlier. Define $\Pi'$ as the deduction obtained by replacing every formula occurrence $p^{A}$ in $\Pi$ by $A$ (atoms $q$ occurring on instances of atomic rules for disjunction and $\bot$-elimination which are not mapped to anything by $\alpha$ are not substituted). Since every instance of every atomic rule becomes some instance of a rule in our system of natural deduction, it is straightforward to show by induction on the length of derivations that $\Pi'$ is a deduction showing $\Gamma \vdash_{\NE} A$.

\medskip

    \item $\mathfrak{B}$ is inconsistent. Then there is a deduction $\Pi$ in $\mathfrak{B}$ showing $\vdash_{\mathfrak{B}} \bot$. If $\Pi$ does not use any rule contained in $\mathfrak{B}$ but not in $\mathcal{N}$, we have $\vdash_{\mathcal{N}} \bot$, contradicting Theorem \ref{thm:atomicconsistencyproof}. Then, $\Pi$ must use some of the new rules. But then we may replace every new rule of $\mathfrak{B}$ which concludes $p^{B}$ by an assumption with shape $p^{B}$ to obtain a deduction showing $\Delta \vdash_{\mathcal{N}} \bot$ for some $\Delta \subseteq \Gamma^{\Atb}$. Define $\Pi'$ as the deduction obtained by replacing every formula occurrence $p^{A}$ in $\Pi$ by $A$. It is straightforward to show by induction on the length of derivations that $\Pi'$ is a deduction showing $\Gamma \vdash_{\NE} \bot$. As a finishing touch, we apply $\bot - elim$ to obtain a deduction showing  $\Gamma \vdash_{\NE} A$.

\end{enumerate}

\end{proof}

Notice that the only non-constructive part of the completeness proof is the step for $A = A^{c}$ in Lemma \ref{lemma:prepcomplete}. This means that by removing $A^{c}$ from the language we would have a fully constructive proof of completeness for intuitionistic logic.

The fact that the inductive steps for constructive proofs only use constructive reasoning but the steps for classical proofs require classical reasoning bears testament to the fact that our definitions indeed capture the meaning of classical and intuitionistic proofs. As such, the ecumenical behaviour observed in the metalanguage should be taken as evidence both of the independence between the distinct notions of proof and of their conceptual adequacy.

\section{Conclusion}

We have proposed a weak and a strong version of $\Bes$ for ecumenical systems. While the first helped furthering our understanding concerning the difference between double negations in intuitionistic logic and provability in classical logic, the ecumenical semantics comes into full swing when the strong notion of $\Bes$ is provided, since it allows for a new ecumenical natural deduction system which is sound and complete w.r.t. it.

%We have proposed a weak and a strong version of $\Bes$ for ecumenical systems. While the first helped furthering our understanding concerning the difference between double negations in intuitionistic logic and provability in classical logic, as well as some relations between local validity, global validity, intuitionistic proofs and classical proofs, the ecumenical semantics comes into full swing when the strong notion of $\Bes$ is provided, since it allows for a new ecumenical natural deduction system which is sound and complete w.r.t. it.

%Our proposal hence provided not only a medium in which classical and intuitionistic {\em logics} may coexist, but also one in which classical and intuitionistic notions of {\em proof} may coexist. This is done through the use of atomic systems and base-extension semantics, which may be used to codify (I) a instuitionistic proof of an atom $p$ as the constructibility of $p$ in the atomic system, and (II) a classical proof of the atom $p$ as the consistency of $p$ with respect to the atomic system.

%\cyan{VN: Esse trecho entre "Our proposal" e "atomic system" ficou excelente, mas acho que ele ficaria melhor l‡ no comeo (talvez no final da "Introduction")! Digo isso porque resume muito bem as ideias filos—ficas principais do texto. Da' o trecho aqui de baixo pode s— ficar "The distinction between classic and intuitionistic notions of proof allow us not only to (...). Da' ficam s— dois par‡grafos e o do meio vai pra introdu‹o}

This distinction allows us not only to obtain classical behaviour for formulas containing classical atoms and intuitionistic behaviour for formulas containing intuitionistic atoms, but also to put on the spotlight basic properties of semantic entailment which are not always evident in traditional semantic analysis. It may also shed light on semantic differences between intuitionistic and classical logics from an even broader perspective.

In the course of this paper we have also shown that it is possible to furnish the absurdity constant $\bot$  with a conceptually adequate and technically sound definition by requiring all systems to be consistent. This can be done in non-ecumenical contexts as well, provided some procedure capable of showing consistency of the syntactic calculus (such as a normalization proof) is available.

There are many ways to further develop this work in the future. First of all, the role of local and global validity in $\Bes$ should be better explored, since it opens wide the classical behaviour as it appears in other semantic settings for classical logic, \eg\ as in Kripke models for classical logic~\cite{DBLP:journals/apal/IlikLH10}. One very interesting step in this direction would be to propose a proof system for our weak version of $\Bes$. Of course, there is the natural question of what would be the $\Bes$ proposal for Prawitz' ecumenical system, from which this work took its inspiration but also other ecumenical systems, such as the ones appearing in~\cite{DBLP:journals/apal/LiangM13,DBLP:journals/Dowek16a,DBLP:journals/lmcs/BlanquiDGHT23}. Another option would be to investigate new combinations of locally and globally defined connectives for the weak semantics. Finally, it would be interesting to lift this discussion to ecumenical modal logics~\cite{DBLP:conf/dali/MarinPPS20}.
\newpage

%\bibliography{references} 

\begin{thebibliography}{29}
\ifx \bisbn   \undefined \def \bisbn  #1{ISBN #1}\fi
\ifx \binits  \undefined \def \binits#1{#1}\fi
\ifx \bauthor  \undefined \def \bauthor#1{#1}\fi
\ifx \batitle  \undefined \def \batitle#1{#1}\fi
\ifx \bjtitle  \undefined \def \bjtitle#1{#1}\fi
\ifx \bvolume  \undefined \def \bvolume#1{\textbf{#1}}\fi
\ifx \byear  \undefined \def \byear#1{#1}\fi
\ifx \bissue  \undefined \def \bissue#1{#1}\fi
\ifx \bfpage  \undefined \def \bfpage#1{#1}\fi
\ifx \blpage  \undefined \def \blpage #1{#1}\fi
\ifx \burl  \undefined \def \burl#1{\textsf{#1}}\fi
\ifx \doiurl  \undefined \def \doiurl#1{\url{https://doi.org/#1}}\fi
\ifx \betal  \undefined \def \betal{\textit{et al.}}\fi
\ifx \binstitute  \undefined \def \binstitute#1{#1}\fi
\ifx \binstitutionaled  \undefined \def \binstitutionaled#1{#1}\fi
\ifx \bctitle  \undefined \def \bctitle#1{#1}\fi
\ifx \beditor  \undefined \def \beditor#1{#1}\fi
\ifx \bpublisher  \undefined \def \bpublisher#1{#1}\fi
\ifx \bbtitle  \undefined \def \bbtitle#1{#1}\fi
\ifx \bedition  \undefined \def \bedition#1{#1}\fi
\ifx \bseriesno  \undefined \def \bseriesno#1{#1}\fi
\ifx \blocation  \undefined \def \blocation#1{#1}\fi
\ifx \bsertitle  \undefined \def \bsertitle#1{#1}\fi
\ifx \bsnm \undefined \def \bsnm#1{#1}\fi
\ifx \bsuffix \undefined \def \bsuffix#1{#1}\fi
\ifx \bparticle \undefined \def \bparticle#1{#1}\fi
\ifx \barticle \undefined \def \barticle#1{#1}\fi
\bibcommenthead
\ifx \bconfdate \undefined \def \bconfdate #1{#1}\fi
\ifx \botherref \undefined \def \botherref #1{#1}\fi
\ifx \url \undefined \def \url#1{\textsf{#1}}\fi
\ifx \bchapter \undefined \def \bchapter#1{#1}\fi
\ifx \bbook \undefined \def \bbook#1{#1}\fi
\ifx \bcomment \undefined \def \bcomment#1{#1}\fi
\ifx \oauthor \undefined \def \oauthor#1{#1}\fi
\ifx \citeauthoryear \undefined \def \citeauthoryear#1{#1}\fi
\ifx \endbibitem  \undefined \def \endbibitem {}\fi
\ifx \bconflocation  \undefined \def \bconflocation#1{#1}\fi
\ifx \arxivurl  \undefined \def \arxivurl#1{\textsf{#1}}\fi
\csname PreBibitemsHook\endcsname

%%% 1
\bibitem[\protect\citeauthoryear{Piecha et~al.}{2015}]{piecha2015failure}
\begin{barticle}
\bauthor{\bsnm{Piecha}, \binits{T.}},
\bauthor{\bsnm{Campos~Sanz}, \binits{W.}},
\bauthor{\bsnm{Schroeder{-}Heister}, \binits{P.}}:
\batitle{Failure of completeness in proof-theoretic semantics}.
\bjtitle{Journal of Philosophical Logic}
\bvolume{44}(\bissue{3}),
\bfpage{321}--\blpage{335}
(\byear{2015})
\doiurl{10.1007/s10992-014-9322-x}
\end{barticle}
\endbibitem

%%% 2
\bibitem[\protect\citeauthoryear{Pym et~al.}{2025}]{pym2023categorical}
\begin{botherref}
\oauthor{\bsnm{Pym}, \binits{D.J.}},
\oauthor{\bsnm{Ritter}, \binits{E.}},
\oauthor{\bsnm{Robinson}, \binits{E.}}:
Categorical Proof-theoretic Semantics
(2025).
\doiurl{10.1007/S11225-024-10101-9} .
\url{https://doi.org/10.1007/s11225-024-10101-9}
\end{botherref}
\endbibitem

%%% 3
\bibitem[\protect\citeauthoryear{Gheorghiu and
  Pym}{2022}]{gheorghiu2022prooftheoretic}
\begin{botherref}
\oauthor{\bsnm{Gheorghiu}, \binits{A.V.}},
\oauthor{\bsnm{Pym}, \binits{D.J.}}:
From Proof-theoretic Validity to Base-extension Semantics for Intuitionistic
  Propositional Logic
(2022).
\doiurl{10.48550/ARXIV.2210.05344} .
\url{https://doi.org/10.48550/arXiv.2210.05344}
\end{botherref}
\endbibitem

%%% 4
\bibitem[\protect\citeauthoryear{Prawitz}{2015}]{DBLP:journals/Prawitz15}
\begin{barticle}
\bauthor{\bsnm{Prawitz}, \binits{D.}}:
\batitle{Classical versus intuitionistic logic}.
\bjtitle{Why is this a Proof?, Festschrift for Luiz Carlos Pereira}
\bvolume{27},
\bfpage{15}--\blpage{32}
(\byear{2015})
\end{barticle}
\endbibitem

%%% 5
\bibitem[\protect\citeauthoryear{Schroeder-Heister}{1991}]{pts-91}
\begin{barticle}
\bauthor{\bsnm{Schroeder-Heister}, \binits{P.}}:
\batitle{Uniform proof-theoretic semantics for logical constants (abstract)}.
\bjtitle{Journal of Symbolic Logic}
\bvolume{56},
\bfpage{1142}
(\byear{1991})
\end{barticle}
\endbibitem

%%% 6
\bibitem[\protect\citeauthoryear{Schroeder{-}Heister}{2006}]{schroeder2006validity}
\begin{barticle}
\bauthor{\bsnm{Schroeder{-}Heister}, \binits{P.}}:
\batitle{Validity concepts in proof-theoretic semantics}.
\bjtitle{Synth.}
\bvolume{148}(\bissue{3}),
\bfpage{525}--\blpage{571}
(\byear{2006})
\doiurl{10.1007/S11229-004-6296-1}
\end{barticle}
\endbibitem

%%% 7
\bibitem[\protect\citeauthoryear{Schroeder-Heister}{2024}]{sep-proof-theoretic-semantics}
\begin{bchapter}
\bauthor{\bsnm{Schroeder-Heister}, \binits{P.}}:
\bctitle{{Proof-Theoretic Semantics}}.
In: \beditor{\bsnm{Zalta}, \binits{E.N.}},
\beditor{\bsnm{Nodelman}, \binits{U.}} (eds.)
\bbtitle{The {Stanford} Encyclopedia of Philosophy},
\bedition{{S}ummer 2024} edn.
\bpublisher{Metaphysics Research Lab, Stanford University},
\blocation{Palo Alto}
(\byear{2024})
\end{bchapter}
\endbibitem

%%% 8
\bibitem[\protect\citeauthoryear{Brandom}{2000}]{Brandom2000}
\begin{bbook}
\bauthor{\bsnm{Brandom}, \binits{R.}}:
\bbtitle{Articulating Reasons: An Introduction to Inferentialism}.
\bpublisher{Harvard University Press},
\blocation{Cambridge}
(\byear{2000})
\end{bbook}
\endbibitem

%%% 9
\bibitem[\protect\citeauthoryear{Sandqvist}{2015}]{Sandqvist2015IL}
\begin{barticle}
\bauthor{\bsnm{Sandqvist}, \binits{T.}}:
\batitle{Base-extension semantics for intuitionistic sentential logic}.
\bjtitle{Logic J. of the IGPL}
\bvolume{23}(\bissue{5}),
\bfpage{719}--\blpage{731}
(\byear{2015})
\doiurl{10.1093/jigpal/jzv021}
\end{barticle}
\endbibitem

%%% 10
\bibitem[\protect\citeauthoryear{Sandqvist}{2009}]{Sandqvist}
\begin{barticle}
\bauthor{\bsnm{Sandqvist}, \binits{T.}}:
\batitle{{Classical logic without bivalence}}.
\bjtitle{Analysis}
\bvolume{69}(\bissue{2}),
\bfpage{211}--\blpage{218}
(\byear{2009})
\doiurl{10.1093/analys/anp003}
{\href{https://arxiv.org/abs/https://academic.oup.com/analysis/article-pdf/69/2/211/296090/anp003.pdf}{{https://academic.oup.com/analysis/article-pdf/69/2/211/296090/anp003.pdf}}}
\end{barticle}
\endbibitem

%%% 11
\bibitem[\protect\citeauthoryear{Makinson}{2014}]{DBLP:journals/igpl/Makinson14}
\begin{barticle}
\bauthor{\bsnm{Makinson}, \binits{D.}}:
\batitle{On an inferential semantics for classical logic}.
\bjtitle{Log. J. {IGPL}}
\bvolume{22}(\bissue{1}),
\bfpage{147}--\blpage{154}
(\byear{2014})
\doiurl{10.1093/jigpal/jzt038}
\end{barticle}
\endbibitem

%%% 12
\bibitem[\protect\citeauthoryear{Dicher and
  Paoli}{2021}]{DBLP:journals/synthese/DicherP21}
\begin{barticle}
\bauthor{\bsnm{Dicher}, \binits{B.}},
\bauthor{\bsnm{Paoli}, \binits{F.}}:
\batitle{The original sin of proof-theoretic semantics}.
\bjtitle{Synth.}
\bvolume{198}(\bissue{1}),
\bfpage{615}--\blpage{640}
(\byear{2021})
\doiurl{10.1007/s11229-018-02048-x}
\end{barticle}
\endbibitem

%%% 13
\bibitem[\protect\citeauthoryear{K{\"{u}}rbis}{2015}]{DBLP:journals/jphil/Kurbis15}
\begin{barticle}
\bauthor{\bsnm{K{\"{u}}rbis}, \binits{N.}}:
\batitle{Proof-theoretic semantics, a problem with negation and prospects for
  modality}.
\bjtitle{J. Philos. Log.}
\bvolume{44}(\bissue{6}),
\bfpage{713}--\blpage{727}
(\byear{2015})
\doiurl{10.1007/s10992-013-9310-6}
\end{barticle}
\endbibitem

%%% 14
\bibitem[\protect\citeauthoryear{Francez}{2016}]{DBLP:journals/logcom/Francez16a}
\begin{barticle}
\bauthor{\bsnm{Francez}, \binits{N.}}:
\batitle{Views of proof-theoretic semantics: reified proof-theoretic meanings}.
\bjtitle{J. Log. Comput.}
\bvolume{26}(\bissue{2}),
\bfpage{479}--\blpage{494}
(\byear{2016})
\doiurl{10.1093/logcom/exu035}
\end{barticle}
\endbibitem

%%% 15
\bibitem[\protect\citeauthoryear{Brouwer}{1981}]{Brouwer1981-BROBCL}
\begin{bbook}
\bauthor{\bsnm{Brouwer}, \binits{L.E.J.}}:
\bbtitle{Brouwer's Cambridge Lectures on Intuitionism}.
\bpublisher{Cambridge University Press},
\blocation{New York}
(\byear{1981})
\end{bbook}
\endbibitem

%%% 16
\bibitem[\protect\citeauthoryear{Dummett}{1977}]{Dummett1977-DUMEOI-2}
\begin{bbook}
\bauthor{\bsnm{Dummett}, \binits{M.}}:
\bbtitle{Elements of Intuitionism}.
\bpublisher{Oxford University Press},
\blocation{New York}
(\byear{1977})
\end{bbook}
\endbibitem

%%% 17
\bibitem[\protect\citeauthoryear{Heyting}{1956}]{Heyting1956-HEYIAI-2}
\begin{bbook}
\bauthor{\bsnm{Heyting}, \binits{A.}}:
\bbtitle{Intuitionism: An Introduction}.
\bpublisher{North-Holland Pub. Co.},
\blocation{Amsterdam,}
(\byear{1956})
\end{bbook}
\endbibitem

%%% 18
\bibitem[\protect\citeauthoryear{Doherty}{2017}]{ConsistencyHillbert61554c58-c869-34f0-b322-2cff263d9ae0}
\begin{botherref}
\oauthor{\bsnm{Doherty}, \binits{F.T.}}:
Hilbert on consistency as a guide to mathematical reality.
Logique et Analyse
(237),
107--128
(2017)
\end{botherref}
\endbibitem

%%% 19
\bibitem[\protect\citeauthoryear{Hilbert et~al.}{1979}]{Hilbert1979-HILMPL}
\begin{barticle}
\bauthor{\bsnm{Hilbert}, \binits{D.}},
\bauthor{\bsnm{Newsom}, \binits{M.W.}},
\bauthor{\bsnm{Browder}, \binits{F.E.}},
\bauthor{\bsnm{Martin}, \binits{D.A.}},
\bauthor{\bsnm{Kreisel}, \binits{G.}},
\bauthor{\bsnm{Davis}, \binits{M.}}:
\batitle{Mathematical problems. lecture delivered before the international
  congress of mathematicians at paris in 1900}.
\bjtitle{Journal of Symbolic Logic}
\bvolume{44}(\bissue{1}),
\bfpage{116}--\blpage{119}
(\byear{1979})
\doiurl{10.2307/2273711}
\end{barticle}
\endbibitem

%%% 20
\bibitem[\protect\citeauthoryear{Hilbert}{1900}]{Hilbert1900}
\begin{barticle}
\bauthor{\bsnm{Hilbert}, \binits{D.}}:
\batitle{\"{U}ber den zahlbegriff.}
\bjtitle{Jahresbericht der Deutschen Mathematiker-Vereinigung}
\bvolume{8},
\bfpage{180}--\blpage{183}
(\byear{1900})
\end{barticle}
\endbibitem

%%% 21
\bibitem[\protect\citeauthoryear{Cobreros}{2008}]{DBLP:journals/sLogica/Cobreros08}
\begin{barticle}
\bauthor{\bsnm{Cobreros}, \binits{P.}}:
\batitle{Supervaluationism and logical consequence: {A} third way}.
\bjtitle{Stud Logica}
\bvolume{90}(\bissue{3}),
\bfpage{291}--\blpage{312}
(\byear{2008})
\doiurl{10.1007/s11225-008-9154-1}
\end{barticle}
\endbibitem

%%% 22
\bibitem[\protect\citeauthoryear{Dummett}{1991}]{dummett1991logical}
\begin{bbook}
\bauthor{\bsnm{Dummett}, \binits{M.}}:
\bbtitle{The Logical Basis of Metaphysics}.
\bpublisher{Harvard University Press},
\blocation{Cambridge}
(\byear{1991})
\end{bbook}
\endbibitem

%%% 23
\bibitem[\protect\citeauthoryear{Nascimento}{2018}]{VictorEcumenical}
\begin{botherref}
\oauthor{\bsnm{Nascimento}, \binits{V.}}:
Logical Ecumenism.
Master's Dissertation.
(2018)
\end{botherref}
\endbibitem

%%% 24
\bibitem[\protect\citeauthoryear{Prawitz}{1965}]{prawitz1965}
\begin{bbook}
\bauthor{\bsnm{Prawitz}, \binits{D.}}:
\bbtitle{Natural Deduction: A Proof-Theoretical Study}.
\bpublisher{Dover Publications},
\blocation{New York}
(\byear{1965})
\end{bbook}
\endbibitem

%%% 25
\bibitem[\protect\citeauthoryear{Ilik
  et~al.}{2010}]{DBLP:journals/apal/IlikLH10}
\begin{barticle}
\bauthor{\bsnm{Ilik}, \binits{D.}},
\bauthor{\bsnm{Lee}, \binits{G.}},
\bauthor{\bsnm{Herbelin}, \binits{H.}}:
\batitle{Kripke models for classical logic}.
\bjtitle{Ann. Pure Appl. Logic}
\bvolume{161}(\bissue{11}),
\bfpage{1367}--\blpage{1378}
(\byear{2010})
\doiurl{10.1016/j.apal.2010.04.007}
\end{barticle}
\endbibitem

%%% 26
\bibitem[\protect\citeauthoryear{Liang and
  Miller}{2013}]{DBLP:journals/apal/LiangM13}
\begin{barticle}
\bauthor{\bsnm{Liang}, \binits{C.C.}},
\bauthor{\bsnm{Miller}, \binits{D.}}:
\batitle{Kripke semantics and proof systems for combining intuitionistic logic
  and classical logic}.
\bjtitle{Ann. Pure Appl. Log.}
\bvolume{164}(\bissue{2}),
\bfpage{86}--\blpage{111}
(\byear{2013})
\doiurl{10.1016/j.apal.2012.09.005}
\end{barticle}
\endbibitem

%%% 27
\bibitem[\protect\citeauthoryear{Dowek}{2016}]{DBLP:journals/Dowek16a}
\begin{barticle}
\bauthor{\bsnm{Dowek}, \binits{G.}}:
\batitle{On the definition of the classical connectives and quantifiers}.
\bjtitle{Why is this a Proof?, Festschrift for Luiz Carlos Pereira}
\bvolume{27},
\bfpage{228}--\blpage{238}
(\byear{2016})
\end{barticle}
\endbibitem

%%% 28
\bibitem[\protect\citeauthoryear{Blanqui
  et~al.}{2023}]{DBLP:journals/lmcs/BlanquiDGHT23}
\begin{botherref}
\oauthor{\bsnm{Blanqui}, \binits{F.}},
\oauthor{\bsnm{Dowek}, \binits{G.}},
\oauthor{\bsnm{Grienenberger}, \binits{{\'{E}}.}},
\oauthor{\bsnm{Hondet}, \binits{G.}},
\oauthor{\bsnm{Thir{\'{e}}}, \binits{F.}}:
A modular construction of type theories.
Log. Methods Comput. Sci.
\textbf{19}(1)
(2023)
\doiurl{10.46298/lmcs-19(1:12)2023}
\end{botherref}
\endbibitem

%%% 29
\bibitem[\protect\citeauthoryear{Marin
  et~al.}{2020}]{DBLP:conf/dali/MarinPPS20}
\begin{bchapter}
\bauthor{\bsnm{Marin}, \binits{S.}},
\bauthor{\bsnm{Pereira}, \binits{L.C.}},
\bauthor{\bsnm{Pimentel}, \binits{E.}},
\bauthor{\bsnm{Sales}, \binits{E.}}:
\bctitle{Ecumenical modal logic}.
In: \bbtitle{DaL{\'{\i}} 2020}.
\bsertitle{LNCS},
vol. \bseriesno{12569},
pp. \bfpage{187}--\blpage{204}.
\bpublisher{Springer},
\blocation{New York}
(\byear{2020}).
\doiurl{10.1007/978-3-030-65840-3\_12}
\end{bchapter}
\endbibitem

\end{thebibliography}
%\bibliographystyle{plain}

\end{document}